\documentclass[12pt]{article}

\usepackage{float}
\usepackage{amsmath}
\usepackage{fullpage}
\usepackage{amsfonts}
\usepackage{amssymb}
\usepackage{graphicx}
\usepackage{caption}
\usepackage{subcaption}
\usepackage{color}
\usepackage{graphics}
\usepackage{epstopdf}
\usepackage{dsfont}
\usepackage{times}
\usepackage{overpic}
\usepackage{framed}
\usepackage{cite}
\usepackage{url}

\usepackage{braket} 		

\usepackage{geometry}
\usepackage{setspace}

\newtheorem{claim}{Claim}
\newtheorem{lemma}{Lemma}

\newtheorem{definition}{Definition}

\numberwithin{equation}{section}
\newenvironment{proof}[1][Proof]{\noindent\textbf{#1.} }{\ \rule{0.5em}{0.5em}}

\renewcommand{\epsilon}{\varepsilon}

\newcommand{\kb}[2]{\ket{#1} \bra{#2}}

\def\N{\mathcal{N}}

\def\Enc{\mathrm{Enc}}
\def\Dec{\mathrm{Dec}}
\def\Ext{\mathrm{Ext}}
\def\Hmin{\mathrm{H}_\mathrm{min}}
\def\Hmax{\mathrm{H}_\mathrm{max}}

\def\R{\mathcal{R}}
\def\TT{\mathcal{T}}
\def\sf{F}					
\def\sfA{F_A}
\def\sfB{F_B}
\def\fail{\mathrm{failure}}
\def\suc{\mathrm{success}}
\def\erase{\mathrm{erase}}

\def\ZO{\{0,1\}}

\makeatletter
\let\@copyrightspace\relax
\makeatother

\begin{document}

\title{Erasable Bit Commitment from Temporary Quantum Trust}

\author{
Norbert L\"{u}tkenhaus \thanks{Institute for Quantum Computing and Department of Physics and Astronomy, University of Waterloo, and Perimeter Institute for Theoretical Physics. Email: lutkenhaus.office@uwaterloo.ca}
\and
Ashutosh Marwah \thanks{Institute for Quantum Computing and Department of Physics and Astronomy, University of Waterloo. Email: ashutosh.marwah@outlook.com} 
\and 
Dave Touchette \thanks{Institute for Quantum Computing and Department of Combinatorics and Optimization,
University of Waterloo, Perimeter Institute for Theoretical Physics, Departement d'informatique and Institut Quantique, Universite de Sherbrooke. Email: touchette.dave@gmail.com
} 
}

\maketitle




\begin{abstract}

We introduce a new setting for two-party cryptography by introducing the notion of temporarily trusted third parties. These third parties act honest-but-curious during the execution of the protocol. Once the protocol concludes and the trust period expires, these third parties may collaborate with an adversarial party. We implement a variant of the cryptographic primitive of bit commitment in this setting, which we call erasable bit commitment. In this primitive, the sender has the choice of either opening or erasing her commitment after the commit phase. For example, she can ask for an erase before the trust period expires in case the conditions for opening the commitment have not been met. The erasure prevents a future coalition of the trusted party and the receiver from extracting any information about the commitment. However, this option also weakens the cryptographic primitive relative to standard bit commitment. Furthermore, the committed information is not revealed to the trusted node at any stage during the protocol. Our protocol requires a constant number of third parties and can tolerate a small number of corrupt third parties as well as implementation errors.

\end{abstract}

\begin{sloppypar}
\section{Introduction}

\subsection{Background}

Bit commitment is a cryptographic protocol between two parties, Alice and Bob, who do not trust each other. The protocol allows Alice to commit to a classical bit in such a way that it is hidden from Bob. Bit commitment is a useful primitive to implement cryptographic tasks such as coin flipping \cite{Blum83} and input swapping \cite{DNS10}. It has also been shown to be equivalent to oblivious transfer \cite{Crepeau94, BF10} and secure multi-party computation \cite{Kilian88}. It is well known that unconditional bit commitment between two parties cannot be implemented classically or even by utilizing quantum resources~\cite{Mayers97, Lo97, Brassard97, DAriano07, Winkler11}. Several different relaxations on the security requirements of bit commitment have been studied in the past. In the classical setting, the most notable of these are settings which require bit commitment to either be statistically binding  and computationally hiding (see, e.g., \cite{Naor91}), or computationally binding and statistically hiding (see, e.g., \cite{Chaum86}). On the other hand, in the quantum setting unconditionally secure bit commitment can be implemented if one assumes that the adversaries have small or noisy quantum memories (see, e.g., ~\cite{DFSS08, DFRSS07, Wehner08, KWW12} and reference therein).\\

Another possible way of implementing bit commitment is to use a trusted third party, who is honest during the protocol. 
(See~\cite{Riv99} for a variant where the trusted party only helps to initialize the protocol.) To commit to a bit $b$, Alice and Bob first share a secret key $k$. Then, Alice sends $(b \oplus k)$ to the third party (through a private channel), who stores the bit until the open phase. During the open phase, Alice simply asks the trusted party to reveal the bit $(b \oplus k)$ to Bob, who is able to extract $b$ from it. Since, the third party is honest, he does not reveal the bit to Bob (Hiding for Bob) and he also does not change the value of the bit after the commit phase (Binding for Alice). Furthermore, the protocol is hiding even for the third party, as he cannot determine $b$ unless someone tells him $k$.\\

Suppose, however, that Alice decides to abandon the commitment instead of opening it. If it is the case that Alice can trust the third party \emph{forever}, Alice's commitment remains a secret from Bob for all future times. However, it is imprudent to base one's trust assumptions on how a party might behave in the future. It is plausible that at some point in the future the third party becomes an adversary to Alice and colludes with Bob after the conclusion of the protocol. If Bob and the third party get together, then they can figure out the value of the bit Alice was committed to.

\subsection{Simple protocol for erasable bit commitment from temporary quantum trust}
\label{subsec:SimpProt}

In fact, due to the no-go theorem\footnote{The reader should note that throughout this paper, the form of security considered is unconditional (information-theoretic) unless otherwise stated.} for classical bit commitment no classical protocol which is binding for Alice can satisfy a hiding property for both Bob and the trusted party together once the protocol is over. On the other hand, one can use quantum mechanics to create a protocol, which only requires Alice and Bob to temporarily trust the third party during the course of the protocol and still to protect the commited information against the third party breaching the trust after the trust period. The basic idea of this protocol, implementing a primitive which we term \emph{erasable bit commitment} (EBC), can be explained by modifying the aforementioned classical protocol. \\

First, Alice and Bob share secret bits $ \theta$ and $k$. In order to commit to bit $b$, Alice sends the state $H^{\theta} \ket{b \oplus k}$ (where $H$ is the Hadamard gate) to the third party in order to commit to $b$. To open the commitment, Alice asks the third party to forward the state to Bob, who can invert the Hadamard depending on the value of the secret bit $\theta$, and open $b$. Furthermore, in case Alice decides to abandon the commitment, she can ask for an "erasure", that is, ask the third party to send her the state back. In this case, once the protocol is over, even if the trusted party and Bob get together, they will not be able to figure out Alice's commitment.  The secret bits Bob shared with Alice were not correlated with Alice's commitment. Whereas, since the third party honestly returned the quantum state back to Alice, he too would not have any information about Alice's commitment. Unlike the classical setting, this primitive is possible in the quantum setting due to the no-cloning theorem. \\

Note that if the trust time is limited, the possibility for Alice to ask for an erasure at the end of the trust period is necessary.
In fact, it could happen that the condition for Alice to issue the open command is not met during the trust period. In such a case, she needs to be able to abort given that the end of the trust period is looming! The erasure command allows her to abort while still ensuring that her data is not revealed to Bob or the third party in the future, even if she can no longer trust the third party anymore. Additionally, the erasure command is announced publicly so that Bob knows that Alice has abandoned her commitment and may not be committed to a classical value anymore. In the rest of the paper, we formalise our security definitions and develop an EBC protocol with multiple trusted third parties, which can tolerate channel noise as well as a small fraction of corrupt third parties.  \\

One might wonder how practical the temporarily trusted party setting is. We envision a scenario where trust is viewed as a resource provided by various independent entities like sovereign states and private companies. For example, if two countries wish to use bit commitment (say in order to exchange information using input swapping), they could consider using a group of other countries as temporarily trusted third parties. The countries acting as third parties will be deterred from acting adversarially because if it is later revealed that these states acted dishonestly, it would be equivalent to breach of sovereign trust and could lead to a weakening of the country's and the government's standing. Similarly, as another example, if two companies in Canada wish to perform bit commitment using temporarily trusted third parties, they could use provincial and municipal governments as the trusted third parties. They could also consider using other independent and competing companies as the trusted third parties. Libracoin \cite{Libracoin}, which is a digital currency initiated by Facebook, uses a similar idea for verifying transactions, wherein it entrusts multiple independent companies to maintain a distributed ledger. In such a situation, the benefit of \emph{temporary} trust becomes evident immediately. Suppose two companies use the companies $\TT_1, \TT_2, \cdots, \TT_m$ as trusted third parties during a commitment protocol, and sometime after the protocol is completed, a single company $C$ acquires a large fraction of the companies $\TT_1, \TT_2, \cdots, \TT_m$. What happens to the security of the commitment scheme in this case? Now, company $C$ will have access to a large fraction of the information available to the companies $\TT_1, \TT_2, \cdots, \TT_m$. It is important that we are able to secure the commitment information in such a scenario as well. For this reason, we consider the temporarily trusted setting. It should be noted that the selection of the temporarily trusted parties completely depends on the past and the present. It is independent of what might happen in the far future. Once a protocol has been securely completed in this setting, its security is maintained irrespective of what happens in the future. In this sense, this setting is "future-proof", like QKD. \\

This paper is organised as follows. We begin by introducing the temporarily trusted party setting in the next section. Then, we describe the erasable bit commitment primitive and the security requirements for such a primitive. Following this we present a robust protocol for this primitive based on BB84 states and prove its security informally. In Section~\ref{sec:AddnSecProp}, we state our claims concerning additional security guarantees. In Section \ref{sec:OTImpossible}, we show that it is impossible to implement a composable protocol for oblivious transfer in this setting. The conclusion discusses how our protocol avoids the no-go results for quantum bit commitment and also compare our protocol to other variants on quantum bit commitment. In the Appendix, we describe our channel noise model, prove the impossibility of classical protocols for erasable bit commitment in a trusted party setting and then formally state definitions and provide security proofs.

\section{Temporarily trusted party setting}

In this section, we describe the temporarily trusted party setting, which we use to implement erasable bit commitment. As stated earlier, the idea is to try to use trusted third parties to implement protocols which cannot be implemented in a two party setting. Further, we do not wish to trust the parties for eternity, we would like to trust them only temporarily for the duration of the protocol. In this section, we formally define what we mean when we use the term \emph{trust}. We also wish to allow for a certain number of dishonest or corrupted trusted parties. We state the trust guarantee in the presence of such parties as well. 

\subsection{Honest-but-curious parties}
\label{subsec:HbCDefn}

We would like our trusted parties to behave in a manner which is indistinguishable from an honest party to an external observer. We call such behaviour \emph{honest-but-curious} behaviour. In particular, the input-output behaviour of honest-but-curious parties is consistent with that of honest participants during the protocol. However, these parties may try to gather as much information as possible during the protocol. We base our definition of $\delta-$honest-but-curious parties on the definition of specious adversaries given in~\cite{DNS10}. The definition considered by Ref. ~\cite{DNS10} is suited for the case of two parties where at most one of those can be honest-but-curious. We generalize this definition to a setting with $p$ number of parties. The parameter $\delta$ is used to quantify the notion of approximately honest-but-curious behaviour. However, we show below that it is sufficient to consider $\delta=0$, or perfectly honest-but-curious behaviour for our protocol. \\

Similar to the previous definition, we define a party to be honest-but-curious if at each step of the protocol it can prove to an external auditor that it has been following the protocol honestly. Specifically, at any stage during the protocol the honest-but-curious party should be able to apply a local map to its state to make its behaviour indistinguishable from that of an honest party. The \emph{local behaviour} part of the definition below captures this requirement. Further, in the \emph{communication behaviour} part of the definition, we also require that the party announce any deviations from the protocol publicly. Note that throughout this paper, we use the notation $\rho \approx_\epsilon \sigma$ for states $\rho$ and $\sigma$ to denote $1/2 \Vert \rho- \sigma \Vert_1 \leq \epsilon$. We also use the notation $[n]$ for the set $\{1,2 \cdots, n\}$.

\begin{definition}[$\delta$-honest-but-curious party] Consider a $k$-step protocol between the parties $(P_n)_{n=1}^p$. Let $\rho_{P_1 \ldots P_p}$ be an initial state distributed between these parties and $\rho_{P'_1 \ldots P'_p R}$ an extension of this state to some reference system $R$. For every strategy of these parties, there is a corresponding quantum channel which is applied by these parties to the initial state. We use this channel to denote the strategy of the parties. In the following, we say that the parties follow a strategy $\Phi^{(j)}_{P'_1 \ldots P'_p}$ if at step $j$ of the protocol the joint state of all the parties is $\Phi^{(j)}_{P'_1 \ldots P'_p} (\rho)$. Further, the honest strategy of party $P_n$ is indicated using $P_n$ in the subscript of $\Phi$ and a general (possibly adversarial) strategy by $P'_n$. \\

\noindent For a $k$-step protocol between the parties $(P_n)_{n=1}^p$, the party $P_i$ is said to be $\delta-$honest-but-curious if 
	\begin{enumerate}
		\item \emph{Local behaviour:} For every strategy of the other parties $P_{\bar{i}} := (P_n)_{n\neq i}$, every initial state $\rho_{P_{i}P_{\bar{i}} R}$ and for every step in the protocol $j \in [k]$, there exists a local CPTP map $\mathcal{T}_j$ such that 
		\begin{align}
			((\mathcal{T}_j)_{P^\prime_{i}} \otimes I_{P^\prime_{\bar{i}}}) \  \Phi^{(j)}_{P^\prime_{i}P^\prime_{\bar{i}}} (\rho_{P_{i}P_{\bar{i}} R} )   \approx_\delta \Phi^{(j)}_{P_{i}P^\prime_{\bar{i}}} (\rho_{P_{i}P_{\bar{i}}R})
		\end{align}
		where $\Phi^{(j)}_{P_{i}P^\prime_{\bar{i}}}$ ($\Phi^{(j)}_{P^\prime_{i}P^\prime_{\bar{i}}}$) is the map applied on the initial state after step $j$ if party $P_i$ follows the protocol honestly (dishonestly) and the other parties follow a general strategy denoted by $P^\prime_{\bar{i}}$ in the channel subscript.
		\item \emph{Communication behaviour:} The party $P_i$ publicly announces if it receives any input or message from any party which deviates from the protocol, in particular if party $P_i$ expects to receive an $n$-dimensional quantum state and receives something lying outside of that space, it publicly announces the deviation.
	\end{enumerate}
\end{definition}

It should be noted at this point that the second requirement above precludes the possibility of the honest-but-curious nodes communicating and collaborating during a protocol. This is different from the definition of classical specious adversaries considered by Ref. \cite{DNS10}, who allow the specious adversaries to collaborate freely. In the scenario considered in Ref. \cite{DNS10}, this was reasonable since they only considered specious adversaries. In our setting, however, we consider both honest-but-curious adversaries and dishonest adversaries. We cannot expect the participants to know beforehand which parties are honest-but-curious and which ones are dishonest. Therefore, we further require that the honest-but-curious nodes do not interact with other parties unless they are required to do so according to the protocol. This assumption, though, can be lifted in our protocol, as we will show in our additional security claim of 'Expungement on successful runs' (Sec. \ref{sec:forget-main}). \\

In Appendix \ref{subsubsec:IntRes}, we show that if a $\delta-$honest-but-curious party is given a quantum state at some point during a protocol, and is asked to return it at a later point in the protocol, we can be certain that the state returned by him was \emph{almost} the same as the one given to him earlier up to tensoring of a fixed state. Specifically, suppose that the honest-but-curious party $T$ is given a state $\rho_{TR}$ during the protocol, and is asked to send his share to party $P$ at a later point. Then, the state $\rho^\prime_{RPT^\prime}$ (where $T^\prime$ is a memory held by $T$) sent by $T$ satisfies
\begin{align*}
	\rho^\prime_{PRT'} \approx_{O(\sqrt{\delta})} (I_R \otimes I_{T \rightarrow P})\rho_{RT} \otimes \tau_{T'}.
\end{align*}
for some fixed state $\tau$. (See Lemma \ref{lemm:HbCDecouple} for formal statement). We see that in this scenario the honest-but-curious party is forced to be almost honest in this scenario. In the main body of this paper, we will therefore assume that the honest-but-curious nodes act completely honestly during the protocol for simplicity, since we will only be dealing with the action of honest-but-curious parties in scenarios like this. Formally, this corresponds to assuming $\delta =0$ or assuming the parties to be perfectly honest-but-curious. This does not lead to loss of generality and the security arguments remain almost the same. A more careful analysis of the behaviour of such parties during our protocol is provided in the Appendix \ref{subsec:SecProofs}. Lastly, since the memory of the honest-but-curious parties after giving away the state $\rho_{RT}$ is almost independent from that state, we say that the honest-but-curious party \emph{forgets} the state he received.\\ 

\subsection{Temporarily trusted party setting}

In addition to the usual two parties, Alice and Bob, involved in a bit commitment protocol, there are also additional trusted third parties $\TT_1, \ldots, \TT_m$ involved in the temporarily trusted party settings we consider. For the purpose of this paper, we refer to all the additional trusted third parties $\TT_1, \ldots, \TT_m$ as trusted parties or trusted nodes. Alice and Bob both trust that out of these $m$ trusted parties at least $(m - t)$ will act honest-but-curiously for the duration of the protocol. In the context of the definition of honest-but-curious above, in our setting the parties are $(A, B, \TT_1, \cdots, \TT_{m})$, and there exists a subset $E \subset [m]$ such that $|E| \leq t$ and for all $j \not\in E$ the party $\TT_j$ is honest-but-curious. Further, the EBC protocol is a two step protocol as will be seen later. The rest of the third parties (at most $t$ in number) can behave dishonestly during the protocol and even collaborate with the cheating party. We use the term 'adversaries' during the protocol to denote the cheating party (dishonest Alice or dishonest Bob) along with the dishonest trusted nodes.\\

Further, this trust assumption is \emph{temporary} or time-limited: after the end of the regular duration of the protocol, Alice and Bob do not have any guarantee about the behaviour of the third parties. They can no longer assume that the $\TT_i$'s behave honestly. These parties could act adversarially after the end of the protocol and even collaborate with a dishonest Alice or a dishonest Bob. 

\section{Erasable bit commitment in a setting with temporarily trusted parties}

Erasable bit commitment (EBC) is a protocol between two parties, Alice and Bob, in a setting with $m$ additional trusted parties, at most $t$ of which are dishonest and adversarial during the protocol. The protocol has two phases: a commit phase, followed by one of either an open or an erase phase. During the commit phase, Alice inputs a string $s$ and Bob receives a message notifying him that Alice has completed the commit phase. At the end of this stage, she is committed to this string, that is, if the commitment is 'opened' then Bob will receive $s$ or reject the commitment. Moreover, the commitment is hidden from Alice's adversaries (Bob and dishonest nodes collaborating with him) at this stage. After the commit phase, Alice can choose to open or erase the commitment. In case of an open, Bob receives the string $s$, the string Alice was committed to. Further, in this case, even if all the trusted parties get together after the protocol they cannot determine $s$. In case Alice chooses to erase, then once the protocol is over, even if Bob and all the trusted parties collaborate they cannot ascertain the value of $s$.\\

In the rest of the paper, we consider a randomized version of erasable bit commitment. Randomized EBC is essentially the same as above except now Alice does not have any input during the commit stage. Instead of choosing the string for her commitment as above, Alice receives a random string $c$ at the end of the commit phase, which is called as her commitment. Our security definitions and proofs are for this randomized version, and we leave direct discussion of non-randomized version for future work. See, e.g., Ref. \cite{KWW12} for a discussion on how to transform a randomized BC into a non-randomized one in their setting by  asking Alice to send $s \oplus c$ to Bob after the commit phase. \\

Thus, in the randomized EBC protocol (referred to as EBC here on) Alice and Bob do not have any inputs during the commit phase. Alice inputs an "open" or "erase" bit at some point after the commit phase into the protocol. In addition to the commitment $c$ mentioned above, Alice and Bob also receive flags $\sfA$ and $\sfB$ (let $\sf=\sfA=\sfB$) notifying them if Alice called for an erase ($\sf= \erase$), or if the open was successful ($\sf = \suc$) or if the open failed ($\sf = \fail$). In general if Alice is dishonest, Bob may not receive the same string as Alice. So, we call Bob's output $\hat{c}$. To summarize, the outputs of the protocol for Alice and Bob are as follows:

\begin{itemize}

\item Alice outputs a string $c \in \{ 0, 1 \}^\ell$ at the end of the commit phase, and a flag $\sfA \in \{ \suc, \fail, \erase \}$ at the end of either the open or erase phase;

\item Bob outputs a string $\hat{c} \in \{ 0, 1 \}^\ell$  as well as a flag $\sfB \in \{ \suc, \fail, \erase \}$ at the end of either the open or erase phase. 

\end{itemize}

\subsection{Security requirements for erasable bit commitment:}
\label{subsec:SecReq}

Given $m$ trusted nodes $(\TT_i)_{i=1}^m$ of which at least $(m- t)$ act honest-but-curiously for the duration of the protocol, an erasable bit commitment protocol satisfies the following security requirements:
\begin{enumerate}
\item Correctness: If both Alice and Bob are honest and Alice receives $c$ during the commit phase, then, in the case of open, Bob accepts the commitment ($\sf = \suc$) and $\hat{c} = c$.

\item Binding: If Bob is honest, then, after the commit phase, there exists a classical variable $\tilde{C}$ such that if Alice chooses to open the commitment then either Bob accepts the commitment ($\sfB = \suc$) and receives an output $\hat{C}$ such that $\hat{C} = \tilde{C}$, or Bob rejects the commitment ($\sfB = \fail$). (This guarantee only makes sense in the case that the open protocol is run after the commit protocol; if an erase protocol is run instead, then Alice does not have to be committed to a classical value anymore.)

\item Hiding: 
\begin{enumerate} \item \emph{Commit:} If Alice is honest and she receives the output $c$ during the commit phase, then, after the commit phase and before the open or erase phase, Bob and the dishonest trusted nodes together can only extract negligible information about $c$, i.e., their joint state is almost independent of Alice's commitment. 

\item \emph{Open:} If Alice and Bob are honest and Alice's commitment is $c$, then, after an open phase, the $\TT_i$'s together can only extract negligible information about $c$, i.e., their joint state is almost independent of Alice's commitment.

\item \emph{Erase:} If Alice is honest and she receives the output $c$ during the commit phase, then, after the erase phase, Bob and the $\TT_i$'s together can only extract negligible information about $c$, i.e., their joint state is almost independent from Alice's commitment. 

\end{enumerate}
\end{enumerate}

We give more formal security definitions based on the security definitions for bit commitment given in Ref. \cite{KWW12} in Appendix \ref{sec:formalSecDefn}. One can easily check that the simple protocol given in Section \ref{subsec:SimpProt} satisfies these requirements (if we consider $m=1$, $t=0$ and $\delta = 0$). \\

Lastly, it should be noted that the binding condition for EBC is weaker than the corresponding binding condition for bit commitment. In EBC, Alice is only committed to a classical value if the commitment is \emph{opened}. This makes EBC a weaker primitive than bit commitment. For example, we cannot use an EBC protocol instead of bit commitment to implement oblivious transfer (OT) using the protocol given in Ref. \cite{Crepeau94}. If we attempt to do so in a straightforward fashion, then the OT protocol created could possibly be susceptible to a superposition attack by Bob. However, we should note that even though EBC is weaker than bit commitment, it cannot be implemented quantum mecahnically in the two-party setting without additional assumptions. The no-go theorem for two-party quantum bit commitment \cite{Mayers97,Lo97} precludes the possibility of a protocol which is both hiding and binding after the commit phase, as EBC is. Finally, we believe that it should be possible to appropriately define tasks like input swapping and coin flipping in the temporarily trusted third party model and then use the protocol given in this paper to implement these, since implementing these tasks using BC requires for all the commitments to be opened. We leave the formal definitions and security proofs for the implementation of these tasks for future work.

\section{Robust Protocol for EBC based on BB84 states}

In this section, we will develop a protocol for EBC using BB84 states which is robust to noise and a small non-zero number of dishonest trusted nodes. This will be done by modifying the simple protocol presented in Section \ref{subsec:SimpProt}. We use tools like privacy amplification and classical error correcting codes for this purpose. Before we proceed further, however, we describe our channel assumptions and our noise model.

\subsection{Channel assumptions}

We assume that the communication between the participants is done over secure (private and authenticated) channels between each pair of participants. We assume that the classical channels are noiseless, but we allow the quantum channels to be noisy. The fact that the quantum channels are noisy somewhat complicates the authentication part of the security guarantee: we cannot simply guarantee that the state output by the channel will be the state which was input. For simplicity, we assume the following, weaker, authentication guarantee with a simple model of adversarial noise: if the quantum channel in one direction between a pair of participants is used $n$ times at some timestep, then at least $(1 - \gamma) n$ of these transmissions will be transmitted perfectly, while the other at most $\gamma n$ transmissions might be arbitrarily corrupted by the adversary. Note that this is sufficient to model some random noise, for example, depolarizing noise, except with negligible probability.\\

Also note that we could use a shared secret key between sender and receiver together with variants of various quantum cryptography techniques, e.g., that of the trap authentication scheme~\cite{BGS13}, to obtain guarantees that with high probability, the actual channel acts as a  superposition over such type of bounded noise channels (in the sense of the adversarial channels of~\cite{LS08}). We leave a detailed analysis of how to implement such a variant of our guarantee in practical settings and the proof of security of our scheme in such settings for future work.
Also note that we do not assume any minimal noise level on the channels; guarantees on minimal noise level can be sufficient to allow one to implement cryptographic primitives like bit commitment (see, e.g., Ref.~\cite{Cre97}).\\

We assume that time-stamping is performed, for example, by broadcasting end-of-time-step signals (with potential abort if some participant is dishonest). As a result, we do not require any synchronicity assumptions. 

\subsection{High-level description}

We consider a setting with $m$ trusted third parties, at most $t$ of which can be dishonest. To understand the tools and the structure of the robust protocol, we try to extend the simple protocol presented in Section \ref{subsec:SimpProt} in a straightforward fashion. Imagine that Alice wishes to commit to the string $x \in \{ 0,1 \}^k$ (this can be selected randomly to perform randomized EBC). To commit to $x$, during the commit phase, Alice randomly selects a basis $\theta \in_R \{ 0,1 \}^k $ and a key $v \in \{ 0, 1 \}^k$,  and distributes the state $ \ket{\psi} = H^\theta \ket{x \oplus v}$ to the $m$-trusted parties, giving $k/m $ qubits to each party. Further, she sends $\theta$ and $v$ to Bob. In case Alice chooses to open, the trusted parties simply forward their shares to Bob and in case she chooses to erase, they send their states back to Alice. If everyone is honest during the protocol and there is no noise, then the string decoded by Bob, $\hat{x}$, is the same as $x$. However, if some qubits sent to Bob are corrupted by the dishonest trusted nodes or noise then $\hat{x} \neq x$. In order to circumvent this problem, we have Alice first encode her string $x$ as $y = \Enc(x)$ using a pre-agreed $[n,k,d]$-classical error correcting code (ECC) with appropriate parameters and then sends it to the trusted parties. Not only would this allow the protocol to be robust to noise, it will also help ensure that the protocol is binding for Alice by making sure that Alice and the adversarial nodes cannot change their data so much as to change the value of the committed string x. Further, if in such a protocol $t$ third parties collaborate with Bob, then they would know $y$ for $\frac{t}{m} n$ positions. This would provide Bob with some partial information about Alice's commitment $x$. Therefore, we require Alice to use privacy amplification (PA) to extract a shorter commitment $c = \Ext (x, r)$, which would be almost independent of her adversaries' share using a randomness extractor $\Ext$. Here $r \in_R \mathcal{R}$ is picked by Alice before the commit phase and sent to Bob during the commit phase. In the next section, we will show how we can choose the right parameters to ensure that the other security requirements are also met. \\

We describe the structure of our protocol before we discuss our parameter choices. Fix an error correcting code with encoding map $\Enc$ and minimum distance $d$, and the randomness extractor $\Ext$. Our protocol has the following structure: \\

\emph{Commit:} Alice randomly selects a $x \in_R \ZO^k$, $z \in_R \ZO^n$, $\theta \in_R \ZO^n$ and $r \in_R \mathcal{R}$. She outputs the random commitment $c = \Ext (x, r)$. Let $y := \Enc(x)$ and $u := y \oplus z $. She then distributes the qubits of $\ket{\psi} := H^\theta \ket{u}$ equally between the trusted nodes and privately sends the classical information $(\theta, z, r)$ to Bob. \\

\emph{Open: } Alice publicly announces 'open' and sends $x$ to Bob privately. The trusted parties forward their shares of $\ket{\psi}$ to Bob. Bob measures $\ket{\psi}$ in $\theta$ basis and checks if the outcome agrees with $x$. If the measured outcome passes the tests, then Bob computes and outputs ${c} := \Ext(x,r)$.\\

\emph{Erase: } Alice publicly announces 'erase'. During the erase phase, the trusted nodes send their shares of $\ket{\psi}$ back to Alice. \\

\begin{figure}
	 \centering
	 
    \begin{subfigure}[htb]{0.4\textwidth}
        \includegraphics[scale=0.2]{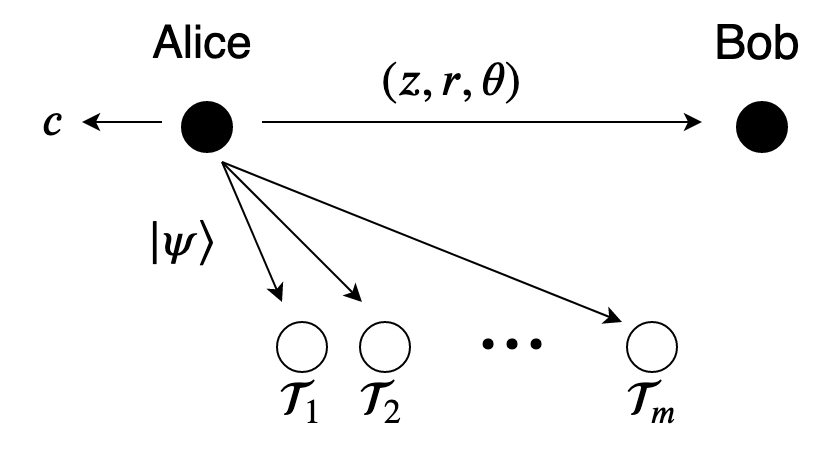}
        \caption{Commit phase}
        \label{fig:commit}
    \end{subfigure}
    \begin{subfigure}[htb]{0.4\textwidth}
    		\centering
        \includegraphics[scale=0.2]{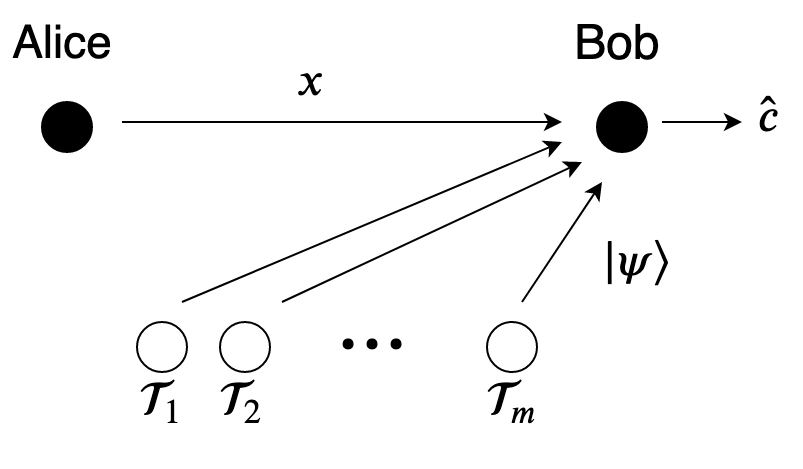}
        \caption{Open phase}
        \label{fig:open}
    \end{subfigure}
    
    \begin{subfigure}[htb]{0.4\textwidth}
    		\centering
        \includegraphics[scale=0.2]{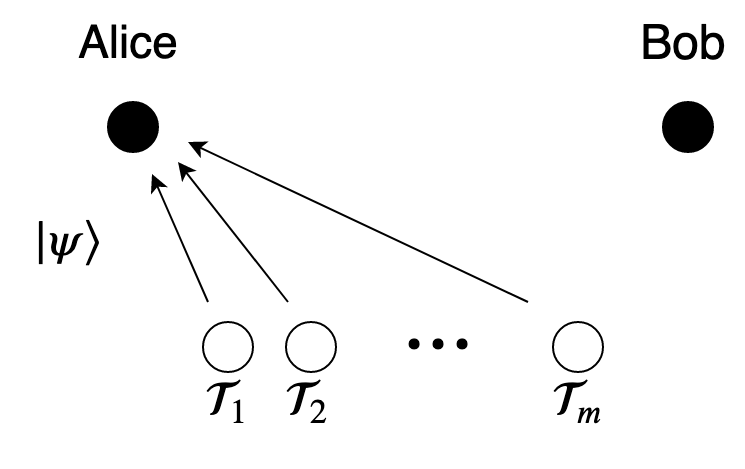}
        \caption{Erase phase}
        \label{fig:erase}
    \end{subfigure}
   \caption{The structure of the protocol is depicted here. During the commit phase, Alice sends $\ket{\psi}$ to the trusted parties and $(z, r, \theta)$ to Bob. In case Alice chooses to open, the quantum state is forwarded to Bob. Whereas, in case of an erase, it is returned back to Alice. }
\end{figure}

\subsection{Choice of parameters}

Consider the protocol given in the structure above. Let's start by looking at the case when both Alice and Bob are honest, and $t$ of the trusted nodes act adversarially. We also account for at most $\gamma n$ of the transmissions being corrupted by the adversaries according to the channel model. Suppose $\tilde{\psi}$ is the state forwarded to Bob by the trusted parties. Let $\hat{u}$ be the string measured by Bob when he measures $\tilde{\psi}$ in the $\theta$ basis and let $\hat{y} := \hat{u} \oplus z$. Since there are at most $t$ dishonest trusted nodes, we have that $h (y, \hat{y}) \leq (t/m + \gamma)n $. So, if Alice encodes $k$ randomly chosen bits $x$ into $y = \Enc (x)$ using a $[n, k, d]$ ECC with minimum distance $d$ satisfying $d > 2 ({t}/{m} + \gamma)n $, then 
$y$ is the only codeword satisfying $h (y, \hat{y}) \leq (t/m + \gamma)n $ and Bob can agree with Alice's $x$ despite $t$ of the trusted nodes being active adversaries. \\

Next, we consider the case when Alice collaborates with these $t$ adversary $\TT_i$'s. In this scenario, is the protocol binding or could Alice hope to change her commitment during the open phase after the commit phase? Suppose that the minimum distance of the ECC is $d = {2 t n}/{m} + 1$. Then, it is possible that Alice could pick $\tilde{y}$ such that there exist $x_1$  and $x_2$ and corresponding codewords $y_1 = \Enc (x_1)$ and $y_2 = \Enc (x_2)$ satisfying $h(\tilde{y}, y_1) = {tn}/{m}$ and $h(\tilde{y}, y_2) = {tn}/{m} + 1$. Alice could then collaborate with the $t$ adversarial trusted nodes to make Bob decode either $x_1$ or $x_2$ at the open phase. In order to avoid this, we require an ECC with a minimum distance $d > 4 ({t}/{m} + \gamma)n$, and only have Bob accept the commitment if $\hat{y}$ is within a distance $({t}/{m} + \gamma)n$ of an actual codeword. Otherwise, Bob outputs "failure" because he is convinced that Alice did not provide a valid commitment. In this way, we can ensure that Alice cannot change the value of the commitment that she opens after the commit phase.\\

Finally, we look at the hiding properties of the protocol. First, observe that after the commit phase Bob and the adversarial trusted nodes have access to at most $(t/m+\gamma)n$ information about $y$. Irrespective of the ECC used this reveals at most $(t/m+\gamma)n$ bits of information about $x$. Hence, we want to ensure that privacy amplification is applied such that $(t/m+\gamma)n$ bits of information about $x$ reveal almost no information about Alice's commitment which is $c = \Ext (x, r)$. For this purpose, we would like the parameter $k$ of the ECC to be large enough compared to $({t}/{m} +\gamma)n$, so that the output length $\ell$ of the randomness extractor used during privacy amplification is $\Theta (n)$ and we can get exponential security in $n$. Therefore, we choose $k = \Theta(n)$, and $\ell = \Theta(n) $ such that $k - (t/m + \gamma)n - \ell = \Omega(n)$. (Here $\Omega(n)$ refers to functions asymptotically lower bounded by $n$ up to constant factors and $\Theta (n)$ to functions both lower and upper bounded by $n$ up to constant factors.)\\

In this way, the information held by the adversary is essentially independent of Alice's commitment $c$ after the commit phase, as well as after the protocol in both the case of an open and an erase. After the protocol, the $(m-t)$ honest-but-curious trusted nodes completely forget the quantum states they held and hence any information about $x$ either by sending their shares back to Alice (erase) or by sending it to Bob (open). As a result, they have no information about $c$ either at the end of the protocol. Furthermore, the information held by the adversaries in both the case of an open and an erase is lesser than $({t}/{m} +\gamma)n$. Thus, even if all the trusted parties and adversaries get together they would be almost independent of Alice's commitment.

\subsection{Protocol}

Given a security parameter $n$, Alice and Bob agree on an $[n, k, d]$ ECC with encoding map $\Enc: \ZO^k  \rightarrow \ZO^n $, as well as a randomness extractor $\Ext: \ZO^k \times \R \rightarrow \ZO^\ell$ for some $\ell$ which we will now define.\\

Let $\gamma$ be the tolerable noise level. We consider a setting with $m$ temporarily trusted nodes at most $t$ of which act dishonestly during the protocol. As discussed above, we fix the various parameters for ECC and PA as follows. For some parameters $\delta_c >0$ and $\delta^\prime >0$:

\begin{itemize}

\item $k = (\frac{t}{m} +  \gamma + \delta_c + \delta^\prime)n$

\item $\ell = \delta_c  n$

\item $d = 4 (\frac{t}{m} + \gamma) n + 1$

\end{itemize}

\emph{Existence of ECC and PA schemes with these parameters:} It should be noted that $(t/m + \gamma)$ should be \emph{small} to allow for the existence of ECC with the above parameters. For example, we can use the Gilbert Varshamov bound~\cite{MS77} to prove that such codes exist for parameters $(t, m, \gamma)$ for which 
\begin{align*}
\frac{t}{m}  + \gamma < 1- H_2\left(4 \left(\frac{t}{m}  + \gamma \right)\right) 
\end{align*}
for large enough $n$ (here $H_2(\cdot)$ is the binary entropy function). This is inequality is satisfied as long as $t/m +\gamma \leq 0.083$. The existence of PA schemes follows from the privacy amplification theorem (see, e.g., Ref.~\cite{Renner05} for an early version robust to quantum side information). \\

\noindent The three phases of the protocol are as follows.

\paragraph{Commit phase}

\begin{enumerate}
	\item Alice randomly picks $x \in_R \ZO^k$, $z \in_R \ZO^n$, $r \in_R \R$, $\theta \in_R \ZO^n$.
	\item Alice computes $c = \Ext (x, r)$, $y = \Enc (x)$, $u = y \oplus z$.
	\item Alice sends $z, r, \theta$ to Bob, who acknowledges reception through broadcast.
	\item Alice prepares $\ket{\psi} = H^\theta \ket{u}$.
	\item Alice distributes the qubits of $\ket{\psi}$ evenly between $\TT_1 \ldots \TT_{m}$ (sends $n/m$ qubits to each trusted party).
	\item $\TT_1 \ldots \TT_{m}$ verify that $\psi$ is in span $\{ \ket{0}, \ket{1} \}^{n}$, and all acknowledge reception through broadcast.
	\item Alice outputs $c \in \ZO^\ell$.
\end{enumerate}


\paragraph{Open phase}

\begin{enumerate}
	\item On input "open", Alice broadcasts "open" and sends $x$ to Bob.
	\item $\TT_1 \ldots \TT_{m}$ send $\psi$ to Bob.
	\item Bob measures $H^{\theta} (\psi)$ in computational basis, gets outcome $\hat{u}$.
	\item Bob decodes $\hat{y} = \hat{u} \oplus z$.
	\item Bob computes $y^\prime = \Enc (x)$  and $h = h (\hat{y}, y^\prime)$.
	\item If $h > ( \frac{t}{m}  +  \gamma ) n$, Bob outputs $\hat{c} = 0^\ell$ and $\sfB = \fail$, else output $\hat{c} = \Ext (x, r)$ and $\sfB = \suc$. Bob broadcasts $\sfB$ and Alice outputs $\sfA = \sfB$.
\end{enumerate}


\paragraph{Erase phase}

\begin{enumerate}
	\item On input "erase", Alice broadcasts "erase".
	\item $\TT_1 \ldots \TT_{m}$ send $\psi$ back to Alice.
	\item Alice measures $H^{\theta} (\psi)$ in computational basis, gets outcome $\hat{u}$.
	\item Alice computes $\hat{y} = \hat{u} \oplus z$ and $h = h (\hat{y}, y)$.
	\item If $h > ( \frac{t}{m}  +  \gamma ) n$, Alice outputs $\sfA = \fail$, else she output $\sfA = \erase$. 
	\item Alice broadcasts $\sfB$ and Bob outputs $\sfB = \sfA$ and $\hat{c} = 0^\ell$.
\end{enumerate}

\section{Security of the protocol}

We now informally prove that our protocol satisfies the security requirements given in Section \ref{subsec:SecReq}. Throughout we assume that we are in a setting with $m$ temporarily trusted third parties $\TT_i$'s out of which at most $t$ can be dishonest during the protocol, and that $\gamma$ is the acceptable rate of noise for the quantum channels. After the completion of the protocol, the trusted parties may bring the information they gathered during the protocol together and collaborate with the adversaries. Once again, we point the reader to Appendix \ref{subsec:SecProofs} for formal security definitions and their proofs.

\subsection{Correctness: Dishonest trusted nodes cannot change the commitment}
\label{sec:corr-open}

\begin{claim}

For honest Alice and honest Bob, the protocol is correct: If Alice outputs $c$ during the commit phase, then Bob outputs $\hat{c} = c$ and $\sfB = \suc$ during the open phase.

\end{claim}

Suppose Alice follows the protocol honestly and prepares $(x, r, z, \theta)$ and $\ket{\psi}$ as required by the protocol. Then, her output is $c= \Ext(x,r)$. Since, at least $(m-t)$ of the trusted nodes are honest-but-curious during the protocol, the $\hat{y}$ measured by Bob satisfies $h (y, \hat{y}) \leq (t/m + \gamma)n$ as the adversarial trusted nodes only have access to at most $(t/m + \gamma)n$ qubits. Thus, Bob accepts Alice's $x$ and outputs $\hat{c}$ which is equal to $c$.


\subsection{Binding for Alice after the commit phase}
\label{sec:bind}

\begin{claim}

For an honest Bob, the protocol is binding for Alice after the execution of the commit phase: If Bob is honest, then, after the commit phase, there exists a classical variable $\tilde{C}$ such that if Alice chooses to open the commitment then either Bob accepts the commitment ($\sfB = \suc$) and receives an output $\hat{C}$ such that $\hat{C} = \tilde{C}$, or Bob rejects the commitment ($\sfB = \fail$). (This guarantee only makes sense in the case that the open protocol is run after the commit protocol; if an erase protocol is run instead, then Alice does not have to be committed to a classical value anymore.)

\end{claim}

Suppose Alice sends some $\tilde{\rho}$ to the trusted nodes and $(\theta, r, z)$ to Bob. Once they are handed over to the trusted nodes, all the qubits of $\tilde{\rho}$, except for at most $({t}/{m} + \gamma)n$ qubits which are controlled by the dishonest $\TT_i$'s and adversarial noise, remain unchanged. Therefore, if we can show that Alice can't change her commitment by modifying any set of $({t}/{m} + \gamma)n$ qubits we will be done. This is achieved by using a classical ECC with distance $d > 4 ({t}/{m} + \gamma) n$ and having Bob reject the received message $\hat{y}$  when there is no codeword in a $({t}/{m} + \gamma) n$ ball around $\hat{y}$. \\

The fact that Alice is committed to a classical variable $\tilde{C}$ after the commit phase if the commitment were to be opened can be seen by using a simulator based argument. Fundamentally, if the open phase were to be conducted right after the commit phase then Bob would measure a random variable $\tilde{Y}$ when he measures $\tilde{\rho}$ in the $\theta$ basis and xors the output with $z$. He could then use $\tilde{Y}$ to figure out Alice's commitment $\tilde{C}$. It will be shown that once Alice commits to a particular instance of this random variable, she cannot change it during the open phase. A formal simulator based argument for this is given in Appendix \ref{subsubsec:BindingProof}.\\

After the commit phase, we define $\tilde{Y}$ as above. Let us consider a particular instance of this random variable $\tilde{y}$. Let $\bar{y}$ be a codeword at a distance less than or equal to $2({t}/{m} + \gamma) n$ away from $\tilde{y}$. Note that there exists at most one such $\bar{y}$   since the minimum distance $d$ satisfies $d = 4({t}/{m} + \gamma) n + 1$. 
If there is no such $\bar{y}$, then even after modifying $({t}/{m} + \gamma) n$ positions of $\tilde{y}$ the distance from any codeword is still greater than  $({t}/{m} + \gamma) n$ and Bob rejects any $x$ Alice tries to open, so let us focus on the case where there exists a unique such $\bar{y}$.
We will show that in this instance Bob either accepts the commitment corresponding to $\bar{y}$ or the protocol ends in failure. Let $\hat{y}$ be as defined in the protocol. Alice's actions can either bring $\hat{y}$ closer to $\bar{y}$ or take it away from it. If her actions manage to get $\hat{y}$ within a distance of $({t}/{m} + \gamma) n$ from $\bar{y}$, then the commitment corresponding to $\bar{y}$ would be accepted by Bob. On the other hand, no matter how she modifies her string, she cannot come within a distance of $({t}/{m} + \gamma) n$ of some other codeword $y^\prime$, since the distance between $y^\prime$ and $\tilde{y}$ would be greater than $2 ({t}/{m} + \gamma) n$ and she only has access to $({t}/{m} + \gamma)n$ qubits. Since for every instance of $\tilde{Y}$ Alice cannot change her commitment after the commit phase, the protocol is binding. 


\subsection{Hiding for Alice against the adversaries after the commit phase}
\label{sec:hid-commit}

\begin{claim}

The protocol is hiding for honest Alice after commit: If Alice is honest and she receives the output $c$ during the commit phase, then, after the commit phase and before the open or erase phase, Bob and the dishonest trusted nodes together can only extract negligible information about $c$, i.e., their joint state is almost indpendent of Alice's commitment. 

\end{claim}

Let $E$ be the set of dishonest trusted parties and $\bar{E}$ the set of honest-but-curious nodes. Together the adversaries  have access to at most $({t}/{m}+\gamma)n$ qubits of $\ket{\psi}$, $z, r$ and $\theta$. Let's say the adversaries' qubits are in the quantum memory $T_E$. Further, suppose that these qubits encode the string $y_E$ (the string $y$ corresponding to the positions of the qubits held by $E$). Then, we have 
\begin{align*}
	\Hmin (X| T_E, Z,R, \Theta) &= \Hmin (Y| T_E, Z,R, \Theta) \\
	& \geq \Hmin (Y| Y_E, Z, R, \Theta) \\
	& = \Hmin (Y| Y_E) \\
	& \geq \Hmin (Y) - |Y_E| \\
	&=\Hmin (X) - |Y_E| \\
	& \geq k - (\frac{t}{m} + \gamma) \\
	& = (\delta_c +\delta^\prime)  n .
\end{align*}
The first equality follows since $Y = \Enc(X)$ and $\Enc$ is a deterministic one-to-one function, the inequality on the second line is a result of applying the data processing inequality, the equality in the third line follows from the fact that $Z,R, \Theta$ are independent of $Y$ and the inequality in the fourth line follows from the chain rule for min-entropy ($|Y_E|$ is the length of the string $Y_E$).\\

Since $ \ell = \delta_c n$, $\Hmin (X | T_E, Z, R, \Theta) - \ell \geq \delta^\prime n$ and by PA theorem, it holds that for any set of adversaries with access limited to these many qubits, the quantum memory $T_E$ is almost independant of the output $c$ of Alice's randomized commitment, even if one is given $z, r$ and $\theta$. Hence, the claim holds true.


\subsection{Hiding for Alice against all the trusted parties and Bob together after the erase phase}
\label{sec:hid-erase}

\begin{claim}

The protocol is hiding for honest Alice after the erase phase: If Alice is honest and she receives the output $c$ during the commit phase, then, after the erase phase, Bob and the $\TT_i$'s together can only extract negligible information about $c$, i.e., their joint state is almost independent of Alice's commitment. 

\end{claim}

At the end of the erase phase, once again all the parties together have at most $({t}/{m} + \gamma)n$ qubits of $\ket{\psi}$. The rest of the qubits are completely forgotten by the $(m-t)$ quantum honest-but-curious nodes, and the content of their register $T_{\bar{E}}$ is independent of $x$. Further, Bob does not get any additional information during the erase phase. Therefore, similar to above, we get 
\begin{align*}
	\Hmin (X| T_E, T_{\bar{E}}, Z, R, \Theta) \geq (\delta_c +\delta^\prime)  n
\end{align*}
and the hiding property holds for the $\TT_i$'s and Bob after the erase phase. 

\subsection{Hiding for Alice and Bob against all the trusted parties together after the open phase}
\label{sec:hid-open-corr}

\begin{claim}

The protocol is hiding for honest Alice and honest Bob after an open phase: If Alice and Bob are honest and Alice's commitment is $c$, then, after an open phase, the $\TT_i$'s together can only extract negligible information about $c$, i.e., their joint state is almost independent of Alice's commitment.

\end{claim}

The information held by all the trusted nodes in this case is lesser than their information in the case of an erase. The claim follows easily by following the argument in Section~\ref{sec:hid-erase}.


\section{Additional security properties}
\label{sec:AddnSecProp}
\subsection{Decoupling of honest-but-curious nodes during the protocol}
\label{sec:HbCDecoupling}

Our protocol satisfies a stronger hiding property than one given in Section \ref{subsec:SecReq}. In addition to being hiding for Bob and the dishonest trusted nodes during the commit stage, the protocol is also hiding for the honest-but-curious trusted nodes in the sense that the qubits they hold locally are also approximately independent of Alice's commitment. A similar strong security definition was used by Ref. \cite{DNS10} for their protocols on secure evaluation of unitaries against honest-but-curious adversaries. In particular, their security definition demands that at any point in the protocol the state held by the honest-but-curious party can be simulated by him by performing local operations on his share of the initial or final state. This definition ensures that the honest-but-curious party never has access to more information than it needs to. It is based on a definition given in Ref. \cite{Watrous02} for statistical zero knowledge proofs. However, this requirement is too strong for cryptographic protocol with temporarily trusted parties, since in the quantum domain one can force the honest-but-curious parties to 'forget' quantum information. So, it is possible that at some point in a protocol a honest-but-curious party holds too much information in its quantum state, but the protocol is such that at a later point the party is forced to return the information and hence 'forgets' it. This is the behaviour that we try to encapsulate in our security definition of erasable bit commitment in a setting with temporarily trusted parties. We only insist that at the end of the protocol, the state held by the honest-but-curious nodes are independent of the commitment. Nevertheless, as mentioned earlier, the protocol presented in this paper satisfies the following stronger security requirement used by Ref. \cite{DNS10}. 

\begin{claim}

For an honest Alice, the protocol is locally hiding for all the honest-but-curious nodes: the local quantum state of the honest-but-curious nodes is almost independent of Alice's output $c$ throughout the protocol.

\end{claim}

This can be proven using the same argument as that of hiding for adveraries (Sec. \ref{sec:hid-commit}). After the commit phase, an honest-but-curious node $\TT_i$ for $i \in \bar{E}$ has access to only $n/m$ qubits in his quantum memory $T_{{\bar{E}_i}}$. Therefore, we once again have the following min-entropy guarantee
\begin{align*}
	\Hmin (X| T_{{\bar{E}_i}}) & \geq \Hmin (X) - \log (|T_{{\bar{E}_i}}|)\\
	& \geq (\delta_c +\delta^\prime)  n.
\end{align*}

\subsection{Expungement on successful runs}
\label{sec:forget-main}

In this section, we show that our protocol has an additional property whenever it is successful, namely that it guarantees that the commitment is indeed 'expunged' from the memories of the trusted nodes, or that the trusted nodes indeed 'forget' about the commitment. In order to prove this, we only need to assume that the trusted nodes did not collaborate with Bob. We do not even need to assume that the trusted nodes acted honestly. The assumption that Bob does not collaborate with the trusted nodes is necessary, since he knows the choice of basis $\theta$, and given the basis, the trusted parties could measure $\ket{\psi}$ without disturbing it. In a sense, all we need to implement our protocol are three sets of parties which do not communicate with each other unless required by the protocol; this is similar to the type of assumptions used in relativistic bit commitment and  we further explore the link between our setting and that of relativistic bit commitment in the conclusion. Hence, it seems that the trust assumption required by our protocol is fairly weak. Further, this security claim also allows us to move closer to the definition of honest-but-curious parties in multi-party protocols given in Ref. \cite{DNS10}, which permits such parties to communicate and collaborate with each other.

\subsubsection{Expungement after the erase phase}

\begin{claim}
For an honest Alice and $\sfA = \erase$, if Bob does not collaborate with the trusted parties during the execution of the protocol, then the protocol is hiding for all the trusted parties and Bob together after the execution of an erase phase, even if the trusted parties collaborate and act arbitrarily during the protocol.
\end{claim}

The setup for this is similar to one for QKD, with Alice after the erase phase acting as "QKD-Bob". Since, we assume that Bob does not collaborate with the trusted parties, they must handle $\ket{\psi} = H^\theta \ket{u}$ and return it to Alice without knowing anything about $\theta$ or $u$. Upon receiving $\psi$ back, Alice measures it according to $\theta$ and only accepts if the measured error rate is at most $\gamma$. This is indeed similar to a QKD setup, and standard methods~\cite{TR11, TLGR12, BF10} allow us to show that, even if $\theta$ is revealed, in the case of an accepted run,
\begin{align}
\Hmin^\epsilon (U | T \Theta) \geq n (1 - H_2 (\gamma + \mu_\epsilon)) - \delta_\epsilon,
\end{align}
for parameters $\mu_\epsilon$, $\delta_\epsilon$, and $T$ the memory held by the trusted parties after the erase phase.
This is not quite sufficient for our purposes, since we want to make sure that the min-entropy about $x$ is high even if, after the protocol, $B$ and the $\TT_i$'s get together, i.e., we want to prove that 
$ 
\Hmin^{\epsilon^\prime} (X | T \Theta Z)
$ 
is large. Using among other tools, the known chain rules for smooth min-entropy~\cite{VDTR13}, we show precisely this in the Appendix \ref{sec:ForgettingProof}. 

Note that we could even allow for Bob and a small set $\TT_E$ of the trusted nodes such that $E \subset [m]$ and $|E| \leq t$ to collaborate with Bob, as long as the remaining $(m-t)$ trusted nodes do not interact with $\TT_E$ and Bob. 

\subsubsection{Expungement after the open phase}

\begin{claim}
For honest Alice and honest Bob, and $\sfA = \sfB = success$ after the execution of an open, the protocol is hiding for all the trusted nodes even if they collaborate and act arbitrarily during the protocol.
\end{claim}

The argument follows similarly as for the previous claim, now with Bob acting as "QKD-Bob", except that the adversary has even less information since we assume Alice and Bob are both honest.

\section{Impossibility of oblivious transfer and strong bit commitment in the temporarily trusted party setting}
\label{sec:OTImpossible}
A natural question is whether we can implement a stronger version of bit commitment in the trusted party setting, where Alice and Bob only need to trust the third parties during the commit phase, and after the commit phase Alice is committed to a classical value in addition to the security requirements presented in Section \ref{subsec:SecReq}. We will refer to such a protocol as \emph{strong bit commitment}. If a composable protocol for oblivious transfer in the trusted party setting exists\footnote{An oblivious transfer protocol has two inputs $x_0, x_1$ for Alice, and one input for Bob $c \in \{0,1\}$. At the end of the protocol, Bob receives the output $x_c$. In the temporarily trusted setting, an OT protocol should be such that Alice and the third parties together can only extract negligible information about Bob's input $c$, and similarly Bob and the third parties together should not be able to extract more than negligible information about $x_{1-c}$.}, then we could use it to implement strong bit commitment using the standard construction of BC from OT \cite{Kilian88}. \\

We can follow the Mayers and Lo-Chau impossibility proofs for two-party bit commitment to show that it is impossible to implement strong bit commitment in this setting. Suppose to the contrary that such a protocol exists. Observe that the trusted parties can purify their actions simply by keeping their auxiliary qubits instead of discarding them. This is valid honest-but-curious behaviour, since an external auditor will not be able to distinguish between the modified protocol and the original protocol. Thus, the new protocol should also have the same security properties as the original protocol. Then, following the original two-party impossibility proofs, Alice and Bob can also purify their actions and the communication. Suppose, Alice commits to the bit $b=0$ in the modified protocol. At the end of the commit phase, the joint state of the parties is $\ket{\psi_{ABT}^{(0)}}$. Since, the protocol is hiding for both Bob and the third parties together after the commit phase (the third parties are only trusted till the end of the commit phase), we would have $\psi_{BT}^{(0)} \approx \psi_{BT}^{(1)}$. This would imply the existence of a unitary $U_A$ on Alice's share, which would allow her to switch between her commitments, i.e., $(U_A \otimes I_{BT}) \ket{\psi_{ABT}^{(0)}} \approx \ket{\psi_{ABT}^{(1)}}$. Further, as we stated above this no-go result for strong-BC also rules out the existence of a composable oblivious transfer protocol in this setting.


\section{Conclusion}

\subsection{Avoiding quantum no-go results}

In this work, we introduced a new primitive, erasable bit commitment, as well as a new paradigm, that of temporary quantum trust.
The temporary trust assumption allows us to avoid no-go results for quantum bit commitment: if Bob and the trusted nodes were to put their information together right after the commit phase, they would be able to extract the value of Alice's commitment. Hence, the protocol is not hiding against a coalition of Bob and all the trusted nodes. In fact, no protocol which is hiding for both Bob and the trusted nodes can be binding for Alice according to the no-go theorem for quantum bit commitment. As already noted, our robust protocol, which makes use of a constant number of trusted nodes, is even robust against a coalition of Bob and a small constant fraction of trusted nodes, as long as the other trusted nodes are indeed acting honest-but-curiously and do not communicate with Bob during the execution period. Note that for security against a dishonest Alice, we cannot guarantee that she is still committed to a classical value at the end of an erase phase. Hence, this primitive is weaker than a standard two-party bit commitment whenever an erase rather than an open might be required.

\subsection{Comparison to other work}

\subsubsection{Bit commitment in bounded and noisy quantum storage model:}

In the bounded and noisy quantum storage model \cite{Damgard05, Wehner08, KWW12}, security of bit commitment is ensured by assuming that the parties do not have access to perfect quantum memory. This assumption allows for the implementation of the stronger, regular version of bit commitment.
In our protocol, for the duration of the trust period, the trusted node must maintain Alice's commitment in memory, but they do not need to further manipulate it. 
In some sense, the requirements here are complementary to those in the noisy storage model. In the noisy storage model, an honest Alice waits for a long enough period for the content of a dishonest Bob's quantum memory to decay sufficiently before announcing her basis. Here, the temporary trust period must be sufficiently short to prevent the contents of the trusted nodes' quantum memories from decaying too much before they are transmitted to Bob. Hence, as progress in physical implementations allows for better and better quantum memory and the power of the noisy storage model decreases, our model allows for longer trust period if so desired.
We believe that efforts towards the physical implementation of this protocol could be an interesting stepping stone towards more complicated multi-node protocols.

\subsubsection{Relativistic bit commitment:}

In relativistic quantum bit commitment \cite{Kent99, Kent11, Kent12, Kaniewski13}, Alice and Bob are split into various nodes at different locations, and the security of the scheme follows from the impossibility of communication between these nodes which is enforced from relativistic constraints.
Hence, Alice has to have complete trust into all of the nodes which are corresponding to her, and similarly for Bob. Contrary to our temporary trust setting, there is no need for nodes which are trusted by both parties. However, the mutual trust between all nodes associated to Alice is usually assumed to be everlasting rather than time-limited as in our setting, and similarly for Bob. Similar to our protocol, relativistic bit commitment is also weaker than bit commitment in the case where the commitment is not opened. 

It would be interesting to see if our protocol can be modified to fit the relativistic framework, by splitting the temporarily trusted nodes into nodes which are temporarily trusted by Alice and nodes which are temporarily trusted by Bob, while keeping only a single main Alice node and a single main Bob node.

\subsubsection{Cheat-sensitive bit commitment:}

In cheat-sensitive bit commitment \cite{Hardy04, BCHKW08, He15}, a dishonest party might be able to cheat with non-trivial probability, but then 
it can also be caught cheating by the honest party with a bounded probability.
In contrast to that, our protocol bounds the probability that a dishonest party trying to cheat is caught and shows that this probability is high.
It might be possible to create stronger versions of cheat-sensitive bit commitment in our trusted party setting. 

\subsubsection{Proving erasure:}

In Ref.~\cite{Unr15}, the authors study revocable quantum time-release encryption. In time-release encryption, one wants that no information about the encrypted message is revealed before some time $T$, while after a fixed time $T'$ such that $T' > T$ the receiver can finally decrypt the message. In revocable quantum time-release encryption, the sender can request back the encryption before time $T$ such that the receiver cannot later gain information to the message. Similar to our erasing procedure, a quantum state is sent back upon a revoking command. Only the setting without trusted third party (even temporarily trusted) is considered in that paper, and accordingly both information-theoretic and complexity-theoretic assumptions are required,  whereas we can prove  information-theoretic security in the temporary trust setting.

In Ref. \cite{CRW19}, the authors investigate whether a server hosting data can information-theoretically prove the deletion of the data. The notion of deletion they use is similar in spirit to our notion of erasure. However, in the setting considered by them, only the case where the client keeps a key much shorter than the message to be hosted and later erased by the server is interesting. Whereas, in our work we are perfectly fine with using a long key. In fact in our main protocol Bob receives a key as long as the message itself. 
Another difference, in the case of an erasure, is that the erasure is validated by a classical message rather than sending back the quantum message as in our case. The work in Ref. \cite{CRW19} was done independently and concurrently with ours.

Another independent and concurrent work is Ref. ~\cite{BI19}. They provide a protocol for a variant of certiﬁed deletion (also with a classical string as proof of deletion) with encryption and prove its security. This might be useful as a subroutine for our application in the trusted third party setting, though no discussion about the commitment property is provided in these works, as it is not the focus there.

\subsubsection{Secure multiparty quantum computation}

In Ref. \cite{Crepeau02}, the authors give a procedure to perform secure multiparty quantum computation (MPQC) as long as the fraction of dishonest parties is less than 1/6, i.e., they show how a set of parties can implement a given quantum circuit on their joint input, without revealing any information about their input to the rest of the parties other than that which can be inferred from their outputs, even if a small subset of the participants is dishonest. We investigate how one might potentially apply their result in the temporarily trusted setting to implement secure two-party computation between Alice and Bob.

Alice, Bob and the input-less temporarily trusted third parties can perform the MPQC protocol. Then, one should be able to implement any quantum circuit between Alice and Bob even if a small subset of the parties is dishonest. If the circuit implements a unitary transform, then no information about Alice's and/or Bob's inputs should remain with the third parties after the computation, since the state held by Alice and Bob (and an inaccessible Reference, if required) after the protocol would be pure. If the circuit does not implement a unitary, then if all the participants purify their actions, the environment system corresponding to the isometric extension of this circuit will be shared between Alice, Bob and the temporarily trusted third parties. In this case, information about Alice's and/or Bob's inputs might potentially leak to the environment system, which could be recovered by the third parties after the protocol. One should formally verify this construction, but it seems that this allows the parties to robustly and securely implement unitary circuits between Alice and Bob with the help of temporarily trusted third parties. One should note that the implementation of unitary extensions of classical functions will not necessarily result in secure implementation of the classical functions themselves. For example, one could implement a unitary extension of oblivious transfer using the MPQC protocol, but this will not result in the same functionality as classical oblivious transfer, since oblivious transfer cannot be implemented in the temporarily trusted setting as we have shown in Section \ref{sec:OTImpossible}. 

The construction used in Ref. \cite{Crepeau02} uses highly complicated quantum states shared between many parties. These states are far from experimental reach. In contrast, our EBC protocol only uses BB84 states. Thus, one might ask if it is possible to use EBC as a simple primitive to implement unitary circuits between Alice and Bob in the temporarily trusted setting.

\subsection{Outlook}

There are many research directions that are sparked off from this work.
A first is to develop EBC protocols which are tailored for optical implementations.
Another is to practically implement a variant of our authentication guarantee for noisy quantum channels from simpler primitives, say shared secret keys, and prove the corresponding security of our scheme in such a model.
It will also be interesting to further study the applications of the erasable bit commitment primitive, and variants. More generally, it will also be interesting to further study the temporary quantum trust setting and what else can be achieved in that setting. Another interesting direction of future research would be to relate this temporary trusted setting with the relativistic setting. In particular, as mentioned above, is there a way to implement an erasable bit commitment, with the help of relativistic constraint, such that some subset of the trusted nodes is only trusted by Alice and the other are only trusted by Bob.

\end{sloppypar}

\section*{Acknowledgement}
We would like to thank the reviewers of QCrypt 2019 and Adrian Kent for their comments and for pointing out the similarity of our setting with relativistic bit commitment, and Sara Zafar Jafarzadeh, Thomas Vidick and Xavier Coiteux-Roy for pointing the connection to Ref. \cite{CRW19}. We are also grateful to the reviewers of IEEE Journal on Selected Areas in Information Theory for their comments. 
We would also like to thank Anne Broadbent for pointing us to her work~\cite{BI19} and also for pointing the connection to Ref.~\cite{Unr15}. This work was performed at the Institute for Quantum Computing, which is supported by Industry Canada and the Perimeter Institute for Theoretical Physics, which is supported in part by the Government of Canada and the Province of Ontario. The research has been supported by NSERC under the Discovery Program, grant number 341495 and by the ARL CDQI program.


\bibliographystyle{unsrt}
\bibliography{erase-BC}

\pagebreak

\pagenumbering{roman}
\section*{Appendix}

\appendix

\section{Channel Model}

As stated in the main body, we study quantum channels with the following authentication guarantee: if the quantum channel in one direction between a pair of participants is used $n$ times in a given timestep, then at least $(1 - \epsilon) n$ of these transmission will be transmitted perfectly, while the other at most $\epsilon n$ transmission might be arbitrarily corrupted by the adversary. 

We now argue that this is sufficient  to model some random noise, e.g., depolarizing noise, except with negligeable probability.

As a model for random channel noise, we use a depolarizing qubit channel, which takes as input a qubit $\rho$ and outputs $(1-\epsilon) \rho$ + $\epsilon \frac{\mathbb{I}}{2}$ for some parameter $\epsilon$. Note that for any $\rho$, $\frac{\mathbb{I}}{2} = \frac{1}{4}(\rho + X \rho X + Y \rho Y + Z \rho Z)$. This is a very general noise model, in particular it is a stronger form of noise than erasure noise, since we can always further decay an erasure channel into a depolarizing channel by outputting $\frac{\mathbb{I}}{2}$ whenever there is an erasure flag. Moreover, if we apply $n$ idenpendent depolarizing channels to $n$ qubits, then, except with probability negligeable in $n$, the output of these channel is a combination of at least $(1 - \frac{3 \epsilon}{4} + \delta) n$ perfectly transmitted qubits plus Pauli errors on the remaining qubits, for some small constant $\delta$. We can then denote $\gamma = 1 - \frac{3 \epsilon}{4} + \delta$ the fraction of perfectly transmitted qubits, except with negligeable probability in $n$.


\section{No-go for classical EBC protocols}
\label{sec:nogo}

In this section, we will show that EBC cannot be implemented classically. The proof of this fact is similar to the proof that unconditionally secure bit commitment cannot be implemented classically. In particular, we show that if the protocol is statistically hiding for Bob and the trusted parties after an erase phase, then the protocol was not binding for Alice in the first place. \\

We assume that all the trusted nodes are honest but curious and that there is no dishonest adversarial node (this is sufficient to prove the result). Moreover, we assume that each party keeps a copy of their message transcripts, since they are allowed to be honest-but-curious. Let $b \in \{ 0, 1 \}$ be Alice's input. The transcripts after each phase of the protocol for honest Alice and Bob are given by:\\

Commit: $\mathcal{M}_b^c := (AB_b^c, AT_b^c, BT_b^c, TT_b^c)$ \\

Open: $\mathcal{M}_b^c \mathcal{M}_b^o := (AB_b^c AB_b^o, AT_b^c AT_b^o, BT_b^c BT_b^o, TT_b^c TT_b^o)$\\

Erase: $\mathcal{M}_b^c \mathcal{M}_b^e := (AB_b^c AB_b^e, AT_b^c AT_b^e, BT_b^c BT_b^e, TT_b^c TT_b^e)$. \\

where the random variable $XY^\text{ph}_b$ represents the transcript of communication between party X and party Y during the phase $\text{ph}$ of the protocol given that Alice's commitment is $b$. The transcript of communication between all the trusted parties and Alice is written as $AT$ for the sake of simplicity. Similarly, $BT$ represents the transcripts between Bob and all of the trusted nodes and $TT$ the transcripts between all of the trusted nodes. Further, let $P_B$ be the random variable of private information and randomness used by Bob and $P_T$ be the random variable of private information and randomness used by all the trusted parties. As we aim to show that Alice can cheat, we assume that Bob follows the honest strategy throughout the protocol. \\

If Bob and the trusted parties get together after an erasure, then all the transcripts would be available to them. The requirement that the protocol be hiding for Bob and the trusted parties after an erasure implies that   
\begin{align*}
	\mathcal{M}_0^c \mathcal{M}_0^e P_B P_T & \approx_\epsilon \mathcal{M}_1^c \mathcal{M}_1^e P_B P_T \\
	\Rightarrow \mathcal{M}_0^c P_B P_T  & \approx_\epsilon \mathcal{M}_1^c P_B P_T 
\end{align*}
using the monotonicity under partial trace. We will now show that a dishonest Alice can use this to cheat. First, she runs the commit phase for $b=0$. Now, if she wishes to open a $0$, she simply runs the open phase honestly. On the other hand, if she wishes to open a $1$, then all she has to do is to run the open phase for $1$. In this case, we have 
\begin{align*}
	\mathcal{M}_0^c P_B P_T & \approx_\epsilon \mathcal{M}_1^c P_B P_T \\
	\Rightarrow \mathcal{M}_0^c \bar{\mathcal{M}}_1^o P_B P_T  & \approx_\epsilon \mathcal{M}_1^c \mathcal{M}_1^o P_B P_T 
\end{align*}
where $\bar{\mathcal{M}}_1^o$ denotes the transcripts produced during the open phase by a dishonest Alice following our procedure. The approximate equality in the second step follows once again from the monotonicity of the trace distance. The right hand side in the second equation above opens $b=1$ due to correctness with high probability ($> 1 -\epsilon$). Thus, except with probability $2\epsilon$, Alice is able to open a $1$ even though she committed to $0$.

\section{Proof of Security for the Erasable Bit Commitment protocol}
\label{sec:formalSecDefn}

\subsection{High-level description}

Our definitions are direct adaptation of those of Ref.~\cite{KWW12} in the two-party setting to the temporarily trusted setting we study here. There, informally, the bit commitment scheme consist of commit and open primitives between two parties, Alice and Bob. First, Alice and Bob execute the commit primitive, where Alice has input $x \in \{ 0, 1 \}^\ell$ and Bob has no input. As output, Bob receives a notification that Alice has chosen an input $x \in \{0, 1 \}^\ell$. Afterwards, they may execute the open protocol, during which Bob either accepts or rejects. If both parties are honest, Bob always accepts and receives the value $x$. If Alice is dishonest, however, we still demand that Bob either outputs the correct value of $x$ or rejects (binding). If Bob is dishonest, he should not be able to gain any information about $x$ before the open protocol is executed (hiding).

Here, for our erasable bit commitment scheme, Alice and Bob use help from temporarily trusted nodes $\TT_1, \ldots, \TT_{m}$ in order to implement the commit and open primitives. We assume that, except for at most $t$ of these nodes which might collaborate adversarially, perhaps with a dishonest Alice or Bob, all these temporarily trusted nodes act (quantum) honest-but-curiously during the protocol, without collaborating with any other nodes during the execution of the protocol. After the execution of the protocol, the trust assumption is lifted: all the temporarily trusted nodes $\TT_1, \ldots, \TT_{m}$ and the dishonest party might gather the data they accumulated during the protocol and act adversarially against the honest parties. In order to complete the execution of the protocol (and lift the corresponding trust assumption) after a commit phase which will not be opened, we add a phase to complement the open phase, which we call an erase phase. This phase forces the temporarily trusted nodes to forget about the commitment before we end the protocol and lift the trust assumption.

As is standard to simplify security definitions and protocols achieving such security, and as is done in Ref.~\cite{KWW12}, we make use of a randomized version of a commitment. Instead of inputting her own string $x $, Alice now receives a random string $c$ at the end of the commit phase.

\subsection{Formal definitions for erasable bit commitment}

To give a more formal definition, we also provide direct adaptation of the definitions of Ref.~\cite{KWW12} to the current setting with temporarily trusted nodes.
Note that we may write the commit, open and erase protocols as CPTPMs $C_{ABT}$, $O_{ABT}$ and $E_{ABT}$, respectively, 
consisting of the local actions of honest Alice, Bob and the temporarily trusted nodes on registers $T = T_1 \ldots T_{m}$, together with any operations 
they may perform on messages that are exchanged.

When all parties are honest, the output of the commit protocol will be a state $C_{ABT} (\rho_{in}) = \rho_{CABT}$ for 
some fixed input state $\rho_{in}$, where $C \in \{0, 1 \}^\ell$ is the classical output of Alice, and $A$, $B$ and $T$ are the internal states of Alice, 
Bob and the temporarily trusted nodes, respectively. If Alice and some nodes $\TT_E$, for $E \subseteq [m]$ and $|E| \leq t$, are dishonest and collaborate,  then they might not follow the protocol, 
and we use $C_{A^\prime B T_E^\prime T_{\bar{E}}}$ to denote the resulting map, where $\bar{E}:=[m] \setminus E$. Note that $C_{A^\prime B T_E^\prime T_{\bar{E}}}$ might 
not have output $C$, hence we simply write $\rho_{A^\prime B T_E^\prime T_{\bar{E}}}$ for the resulting output state, where $A^\prime T_E^\prime$ denote the registers 
of the dishonest Alice and $\TT_E$. Similarly, we use $C_{A B^\prime T_E^\prime T_{\bar{E}}}$ to denote the CPTPM corresponding to the case where Bob and $\TT_E$ are dishonest, 
and write $\rho_{C A B^\prime T_E^\prime T_{\bar{E}}}$ for the resulting output state, where $B^\prime T_E^\prime$ denote the registers of the dishonest Bob and $\TT_E$. 
The case of both Alice and Bob honest but dishonest $\TT_E$ is similar.

 The open protocol can be described similarly. If all parties are honest, the map $O_{AB T}$ creates 
the state $\eta_{C \tilde{C} F} : = (I_{C} \otimes O_{ABT}) (\rho_{C ABT})$, where $\tilde{C} \in \{ 0, 1 \}^\ell$ 
and $F \in \{\suc , \fail \}$ is the classical output of Bob. Again, if Alice and $\TT_E$ are dishonest,  
we use $O_{A^{\prime} B T_E^{\prime } T_{\bar{E}}}$ to denote the resulting 
map with output $\eta_{A^{\prime}  T_E^{\prime}  T_{\bar{E}} \tilde{C} F}$. Similarly, 
we use $O_{C A B^{\prime } T_E^{\prime } }$ to denote the 
CPTPM corresponding to the case where Bob and $\TT_E$ are dishonest, 
and write $\eta_{C  B^{\prime } T_E^{\prime} }$ for the resulting output state.
In the case of both Alice and Bob honest but dishonest $\TT_E$, we use $O_{C A B T_E^{\prime } T_{\bar{E}} }$ 
to denote the CPTPM corresponding to the case  $\TT_E$ are dishonest, and write $\eta_{C \tilde{C} F  T_E^{ \prime } T_{\bar{E}} }$ 
for the resulting output state
also considering the memory content of the quantum honest-but-curious nodes in $\TT_{\bar{E}}$.

The erase protocol can also be described similarly. If all parties are honest, the map $E_{AB T}$ creates 
the state $\zeta_{C  F} : = (I_{C} \otimes E_{ABT}) (\rho_{C ABT})$, where $F = \text{erase} $ is the classical output of Bob. Again, if Alice and $\TT_E$ are dishonest,  
we use $E_{A^{\prime } B T_E^{ \prime } T_{\bar{E}}}$ to denote 
the resulting map with output $\zeta_{A^{\prime }  T_E^{\prime }  F}$. 
Similarly, we use $E_{C A B^{\prime} T_E^{\prime } T_{\bar{E}} }$ to denote 
the CPTPM corresponding to the case where Bob and $\TT_E$ are dishonest, 
and write $\zeta_{C  B^{\prime} T_E^{\prime}  T_{\bar{E}} }$ for the 
resulting output state
also considering the memory content of the quantum honest-but-curious nodes in $\TT_{\bar{E}}$. 
The case of both Alice and Bob honest but dishonest $\TT_E$ is similar. Finally, we use $\tau_{X}$ to denote the state, which is uniformly random over the set $X$, i.e., $\tau_X := \sum_{x \in X} |X|^{-1} \ket{x}\bra{x} $.

The security definitions for EBC as defined here are direct adaptations of the definitions in Ref. \cite{KWW12}, which are themselves a generalization of the ones in Ref. \cite{DFRSS07}, to our setting.

\begin{definition}
\label{def:ebc}

An $(\ell, \epsilon)$ randomized erasable string commitment scheme is a protocol between Alice, Bob and temporarily trusted nodes $\TT_1, \ldots, \TT_{m}$ satisfying the following properties:

\textbf{Correctness:}

Open: If both Alice and Bob are honest, the set of dishonest $\TT_i$, $E \subset [m]$ is such that $|E| \leq t $, and the rest of the $\TT_i$ are quantum honest-but-curious then, after the open protocol, the ideal state $\sigma_{CCFT}$ is defined such that

\begin{enumerate}

\item The distribution of $C$ is uniform, Bob accepts the commitment and the $\TT_i$'s learn nothing about $c$:

$$\sigma_{CFT} = \tau_{\{ 0, 1 \}^\ell } \otimes \kb{\suc}{\suc} \otimes \sigma_T .$$

\item The joint state $\eta_{C \tilde{C} F T}$ created by the real protocol is $\epsilon$-close to the ideal state

$$ \eta_{C \tilde{C} F T} \approx_\epsilon \sigma_{CCFT} .$$

\end{enumerate}

Erase: If both Alice and Bob are honest, $|E| \leq t$ and the rest of the $\TT_i$ are quantum honest-but-curious, then after the erase protocol, the ideal state $\sigma_{C \tilde{C} F T}$ is defined such that

\begin{enumerate}

\item The distribution  of $C$ is uniform, that of $\tilde{C}$ is constant,  Bob outputs an erasure flag and the $\TT_i$'s learn nothing about $c$:

$$\sigma_{C \tilde{C} F T} = \tau_{\{ 0, 1 \}^\ell } \otimes \kb{0}{0} \otimes \kb{erase}{erase} \otimes \sigma_T$$

\item The joint state $\zeta_{C \tilde{C} F T}$ created by the real protocol is $\epsilon$-close to the ideal state

$$ \zeta_{C \tilde{C} F T} \approx_\epsilon \sigma_{C\tilde{C}FT} . $$

\end{enumerate}

\textbf{Security for Alice ($\epsilon$-hiding):}

Commit: If Alice is honest, $|E| \leq t$ and $\TT_{\bar{E}}$ are quantum honest-but-curious, then for any joint state $\rho_{C B^\prime T_E^\prime T_{\bar{E}}}$ created by the commit protocol, Bob and the adversarial nodes do not learn $C$:

$$\rho_{C  B^\prime T_E^\prime } \approx_\epsilon \tau_{\{ 0, 1 \}^\ell } \otimes \rho_{  B^\prime T_E^\prime } $$

Erase: If Alice is honest, $|E| \leq t$ and $\TT_{\bar{E}}$ are quantum honest-but-curious, then for any joint state $\zeta_{C  B^{\prime } T_E^{\prime  }  T_{\bar{E}} }$ created by the erase protocol, the adversaries together with the honest-but-curious nodes do not learn $C$:

$$\zeta_{C  B^{\prime  } T_E^{\prime  }  T_{\bar{E}} } \approx_\epsilon \tau_{\{ 0, 1 \}^\ell } \otimes \zeta_{B^{\prime  } T_E^{\prime }  T_{\bar{E}} } . $$

\textbf{Security for Bob ($\epsilon$-binding):}

If Bob is honest, $|E| \leq t$ and $\TT_{\bar{E}}$ are quantum honest-but-curious then there exists an ideal state $\sigma_{\tilde{C} A^\prime B T_E^\prime T_{\bar{E}}}$ where $\tilde{C}$ is a classical register such that for all open maps $O_{A^{\prime } B T_E^{\prime  } T_{\bar{E}} \rightarrow A^{\prime } \hat{C} F T_E^{\prime  } T_{\bar{E}}}$

\begin{enumerate}

\item Bob almost never accepts $ \hat{C} \not= \tilde{C} $:

$$\mathrm{For}~\sigma^O_{\tilde{C} A^{\prime }  \hat{C} F T_E^{\prime } T_{\bar{E}}} = (I_{\tilde{C}} \otimes O ) (\sigma_{\tilde{C} A^\prime B T_E^\prime T_{\bar{E}}}) , \mathrm{we~have}$$
$$\mathrm{Pr}[\hat{C} \not= \tilde{C}~\mathrm{and}~F=\suc ~] \leq \epsilon . $$

\item If the open map is applied to the real committed state, the joint state produced $\eta_{A^\prime \hat{C} F T_E^\prime T_{\bar{E}}} := O (\rho_{ A^\prime B T_E^\prime T_{\bar{E}}})$ is $\epsilon$-close to the ideal state

$$  \eta_{A^\prime \hat{C} F T_E^\prime T_{\bar{E}}} \approx_\epsilon \sigma^O_{A^\prime \hat{C} F T_E^\prime T_{\bar{E}}} . $$

\end{enumerate}

\end{definition}

Since the temporarily trusted party setting involves multiple parties, one might wish to use a simulator based security definition, where a protocol is said to be secure if the adversaries' state in the real protocol could be extracted from a simulator acting on the ideal protocol. The security definition given above internalises this reduction to the ideal protocol and the subsequent simulation in the security definition itself, that is, the real state of the protocol is compared to an ideal state which would otherwise have been produced by an ideal functionality and then processed by a simulator. We can show the equivalence of these two approaches by considering the ideal functionality for bit commitment in this setting directly as well. After the commit phase, the ideal functionality outputs a uniformly random variable to Alice and gives out trivial outputs to all the other parties. In the real protocol, the state held by Bob and the dishonest trusted third parties are almost independent of Alice's commitment according to the definition given above. Since, these states are independent of Alice's state, a simulator could have produced these from the trivial outputs given out by the ideal functionality. Further, we also show in Sec. \ref{sec:HbCDecoupling} that the states held by each of the honest-but-curious trusted third parties is almost independent of Alice's commitment, which means that these too could have been produced using a simulator (This is not required in our security definition of EBC, but this is satisfied by our protocol nevertheless. It might be interesting to investigate whether the stronger claim provides any advantage operationally or from the point of view of composability.). Similarly, after the erase phase one can show that the outputs of the third parties and Bob can be simulated from the ideal functionality. In the case of a dishonest Alice, the ideal functionality guarantees that for every strategy of Alice and the dishonest third parties, the state produced after the commit phase (which has been referred to as the ideal state in the definition above) is such that there exists a fixed random variable which would be output to Bob if the open phase were to be executed successfully. Our security definition requires that the real state of the protocol after the open phase is close to an ideal state after the open phase.

\subsection{Security proofs}
\label{subsec:SecProofs}
\subsubsection{Intermediate results}
\label{subsubsec:IntRes}

We will prove some Lemmas, which will act as intermediaries to prove security. In particular, we will analyze the actions of the honest-but-curious nodes in the protocol. 

\begin{lemma}[Corollary to Uhlmann's Theorem]
	Suppose that $\rho, \rho^\prime \in D(\mathcal{X})$ such that ${1/2 \Vert \rho - \rho^\prime \Vert_1} < \epsilon$. Then, there exist purifications of $\rho $ and $\rho^\prime $, $\ket{\psi}$ and $\ket{\psi^\prime}$ in $\mathcal{X} \otimes \mathcal{X}$ respectively such that 
	\begin{align}
		\ket{\psi} \bra{\psi} \approx_{\sqrt{2 \epsilon}} \ket{\psi^\prime} \bra{\psi^\prime} 
		\label{eq:UhlmannsTh}
	\end{align}
	\label{lemm:UhlmannsTh}
\end{lemma}
\begin{proof}
	Using the Fuchs-van de Graaf inequality we have that
	\begin{align*}
		F(\rho, \rho^\prime) \geq 1-\frac{1}{2} \Vert \rho - \rho^\prime \Vert_1 > 1- \epsilon.
	\end{align*}
	Further using Uhlmann's theorem we have that there exist purifications $\ket{\psi}$ and $\ket{\psi^\prime}$ in $\mathcal{X} \otimes \mathcal{X}$ of $\rho$ and $\rho^\prime$ such that
	\begin{align*}
		F(\rho, \rho^\prime) = \vert \braket{\psi | \psi^\prime} \vert, 
	\end{align*}
	which implies that 
	\begin{align*}
		\vert \braket{\psi | \psi^\prime}  \vert > 1 - \epsilon.
	\end{align*}
	Now using the fact (\cite[Eq. 1.186]{Wat18}) that 
	\begin{align*}
		\frac{1}{2} \Vert \ket{\psi} \bra{\psi} - \ket{\psi^\prime} \bra{\psi^\prime} \Vert_1 = \sqrt{1 - \vert \braket{\psi | \psi^\prime} \vert^2 }
	\end{align*}
	we have
	\begin{align*}
		\frac{1}{2} \Vert \ket{\psi} \bra{\psi} - \ket{\psi^\prime} \bra{\psi^\prime} \Vert_1 < \sqrt{2 \epsilon}.
	\end{align*}
\end{proof}

In the following Lemma, we consider the action of an honest-but-curious node in the erasable bit commitment protocol in a general fashion. We suppose that Alice sends the $T$ part of the state $\rho_{RT}$ to the party $T$. The $R$ part of the state may be arbitrarily distributed between the rest of parties. For example, in the real protocol the register $R = (A, \bar{T})$, where $\bar{T}$ represents all the trusted nodes other than the node $T$. Then, we consider the cases of 'Open' and 'Erase' and show that the state held by $T$ after these are almost decoupled from the state that Alice gave him. In particular, we assume that after the announcement of 'Open' (or 'Erase'), $T$ applies the map $\Phi^O_{T \rightarrow BT^\prime}$ (or $\Phi^E_{T \rightarrow AT^\prime}$) to his part of $\rho_{RT}$ and sends the $B$ (or $A$) share of the state to Bob (or Alice). Similarly, the rest of the parties apply an arbitrary channel, $\Phi_{R}$ to the register $R$. This channel can be arbitrary since the behaviour of $T$ should appear honest irrespective of the strategies of the other parties by definition. We then show that 
\begin{align*}
	(\Phi_{R} \otimes \Phi^O_{T \rightarrow BT^\prime}) (\rho_{RT}) \approx_{O(\sqrt{\delta})} (\Phi_{R} \otimes I_{T \rightarrow B}) (\rho_{RT}) \otimes \tilde{\sigma}_{T^\prime} \\
	(\Phi_{R} \otimes \Phi^E_{T \rightarrow AT^\prime}) (\rho_{RT}) \approx_{O(\sqrt{\delta})} (\Phi_{R} \otimes I_{T \rightarrow A}) (\rho_{RT}) \otimes \tilde{\omega}_{T^\prime}
\end{align*}
where $\tilde{\sigma}_{T^\prime}$ and $\tilde{\omega}_{T^\prime}$ are fixed states. 

\begin{lemma}
	Let $T$ be a $\delta-$honest-but-curious node in the erasable bit commitment protocol. Suppose that during the commit stage of the protocol, Alice sends the $T$ substate of the state $\rho_{RT}$ to the party $T$. The $R$ part of the state is distributed arbitrarily between the parties other than $T$. Further, suppose that in the case of an 'Open', $T$ applies the channel $\Phi^O_{T \rightarrow BT^\prime}$ to his state $T$. Then, there exists a state $\tilde{\sigma}_{T^\prime}$ such that for every state $\rho_{RT}$ distributed by Alice during the commit stage and every channel $\Phi_R$ applied on $R$ during the open phase, the state $\eta_{RBT^\prime} := (\Phi_{R} \otimes \Phi^O_{T \rightarrow BT^\prime}) (\rho_{RT}) $ (where the $B$ part is given to Bob by $T$ and the $T^\prime$ part is held by $T$) produced during the open phase is such that 
	\begin{align*}
		\eta_{RBT^\prime} \approx_{5\sqrt{2 \delta}} ( \Phi_{R} \otimes I_{T \rightarrow B} ) (\rho_{RT}) \otimes \tilde{\sigma}_{T^\prime}.
	\end{align*}
 Similarly, suppose that in the case of an 'Erase', $T$ applies the channel $\Phi^E_{T \rightarrow AT^\prime}$ to his state $T$. Then, there exists a state $\tilde{\omega}_{T^\prime}$ such that for every state $\rho_{RT}$ distributed by Alice during the commit stage and every channel $\Phi_R$ applied on $R$ during the erase phase, the state $\zeta_{RAT^\prime} := (\Phi_{R} \otimes \Phi^E_{T \rightarrow AT^\prime}) (\rho_{RT}) $ (where the $A$ part is given back to Alice by $T$ and the $T^\prime$ part is held by $T$) produced during the erase phase is such that 
	\begin{align*}
		\zeta_{RAT^\prime} \approx_{5\sqrt{2 \delta}}  ( \Phi_{R} \otimes I_{T \rightarrow A} ) (\rho_{RT}) \otimes \tilde{\omega}_{T^\prime}.
	\end{align*}
	\label{lemm:HbCDecouple}
\end{lemma}

\begin{proof}
	The proof of this Lemma follows the same argument as the proof of the Rushing Lemma in (Lemma 6.2 in \cite{DNS10}). We will prove this Lemma only for the case of an 'Open'. The case of an 'Erase' follows similarly. Note that due to convexity of the trace norm, it suffices to prove this Lemma for pure states. Further, it is also sufficient to consider isometric channels $\Phi_R$ for an arbitrary register $R$, since one can use the Stinespring representation to write a general channel as an isometric channel composed with the partial trace (see for example \cite[Sec. 2.2.2]{Wat18}) and the trace distance decreases under a partial trace operation.  \\
	
	Suppose that after the announcement of 'Open', the trusted node $T$ applies the map $\Phi^O_{T \rightarrow B T^\prime}$ and the isometric channel $\Phi_R$ is applied to the register $R$ by the other parties. Choose and fix a state $\rho^{(1)}_{RT}:= \ket{\psi^{(1)}_{RT}} \bra{\psi^{(1)}_{RT}}$. We then define the state
	\begin{align*}
		\eta^{(1)}_{RBT^\prime} = (\Phi_R \otimes \Phi^O_{T \rightarrow B T^\prime} )(\rho^{(1)}_{RT}).	
	\end{align*}
	Since $T$ is $\delta$-honest-but-curious in the case of an open, there exists a channel $\mathcal{T}_{T^\prime}$ such that 
	\begin{align*}
		(I_{RB} \otimes \mathcal{T}_{T^\prime}) (\eta^{(1)}_{RBT^\prime})  & \approx_{\delta} (\Phi_R \otimes I_{T \rightarrow B})(\rho^{(1)}_{RT})\\
		\Rightarrow \ \eta^{(1)}_{RB} & \approx_{\delta} (\Phi_R \otimes I_{T \rightarrow B})(\rho^{(1)}_{RT}).
	\end{align*}
	Now, we use Lemma \ref{lemm:UhlmannsTh} and the fact that $(\Phi_R \otimes I_{T \rightarrow B})(\rho^{(1)}_{RT})$ is a pure state to show that there exists a state ${\sigma}^\prime_{T^{\prime\prime}}$ and a purification $\eta^{(1)}_{RBT^{\prime\prime}}$ of $\eta^{(1)}_{RB} $ such that 
	\begin{align*}
		\eta^{(1)}_{RBT^{\prime\prime}} \approx_{\sqrt{2\delta}} (\Phi_R \otimes I_{T \rightarrow B})(\rho^{(1)}_{RT}) \otimes {\sigma^\prime}_{T^{\prime\prime}}.
	\end{align*}
	Further, since $\eta^{(1)}_{RBT^{\prime\prime}}$ is a purification of $\eta^{(1)}_{RB}$ and has the same partial state as $\eta^{(1)}_{RBT^{\prime}}$ on registers R and B, there exists a quantum channel $\Omega_{T^{\prime\prime} \rightarrow T^\prime}$ acting only on $T^{\prime\prime}$ transforming $\eta^{(1)}_{RBT^{\prime\prime}}$ into $\eta^{(1)}_{RBT^{\prime}}$ (\cite[Proposition 2.29]{Wat18}). Thus, we have
	\begin{align}
		\eta^{(1)}_{RBT^{\prime}} \approx_{\sqrt{2\delta}} (\Phi_R \otimes I_{T \rightarrow B})(\rho^{(1)}_{RT}) \otimes \tilde{\sigma}_{T^{\prime}}
		\label{eq:PfSigmaTilde}
	\end{align}
	where $\tilde{\sigma}_{T^{\prime}} = \Omega_{T^{\prime\prime} \rightarrow T^\prime} ({\sigma^\prime}_{T^{\prime\prime}})$.\\
	
	We claim that this state $\tilde{\sigma}_{T^{\prime}}$ is the required state. Consider the state $\rho = \ket{\psi_{RT}} \bra{\psi_{RT}}$. Now, let $\ket{\psi^\prime_{R^\prime R T}} := 1/\sqrt{2} (\ket{1_{R^\prime}} \ket{\psi_{RT}} + \ket{2_{R^\prime}}\ket{\psi^{(1)}_{RT}})$ where $\ket{1_{R^\prime}}$ and $\ket{2_{R^\prime}}$ are orthonormal states. Define the states
	\begin{align*}
		\eta_{RBT^\prime} & := (\Phi_R \otimes \Phi^O_{T \rightarrow BT^\prime})(\ket{\psi_{RT}} \bra{\psi_{RT}}) ) \\
		\eta^\prime_{R^\prime R BT^\prime} & := (I_{R^\prime} \otimes \Phi_R \otimes \Phi^O_{T \rightarrow BT^\prime})(\ket{\psi^\prime_{R^\prime R T}} \bra{\psi^\prime_{R^\prime R T}}) ) 	.
	\end{align*} 
	Since the node is $\delta$-honest-but-curious irrespective of the strategy of the other parties, we once again have
	\begin{align*}
		\eta^\prime_{R^\prime R B} \approx_\delta (\Phi_R \otimes I_{T \rightarrow B}) (\ket{\psi^\prime_{R^\prime R T}} \bra{\psi^\prime_{R^\prime R T}}) )
	\end{align*}
	where we have omitted the identity map on $R^\prime$ for brevity. Arguing as before, we see that there exists a state $\sigma^{\prime \prime}_{T^\prime}$ such that 
	\begin{align*}
		\eta^\prime_{R^\prime R B T^\prime} \approx_{\sqrt{2\delta}} (\Phi_R \otimes I_{T \rightarrow B}) (\ket{\psi^\prime_{R^\prime R T}} \bra{\psi^\prime_{R^\prime R T}}) \otimes \sigma^{\prime \prime}_{T^\prime}.
	\end{align*}
	If we now apply the channel of the measurement in the $\{ \ket{1_{R^\prime}}\bra{1_{R^\prime}}, \ket{2_{R^\prime}}\bra{2_{R^\prime}} \}$ to both the sides of the above equation we get
	\begin{align}
		\eta_{RBT^\prime} & \approx_{2 \sqrt{2\delta} } (\Phi_R \otimes I_{(T \rightarrow B)})(\ket{\psi_{RT}} \bra{\psi_{RT}}) ) \otimes \sigma^{\prime \prime}_{T^\prime} \label{eq:PfSigmaArb}\\
		\eta^{(1)}_{RBT^\prime} & \approx_{2 \sqrt{2\delta} } (\Phi_R \otimes I_{(T \rightarrow B)})(\ket{\psi^{(1)}_{RT}} \bra{\psi^{(1)}_{RT}}) ) \otimes \sigma^{\prime \prime}_{T^\prime}  \label{eq:PfSigmaT1Link}
	\end{align}
	Combining Eq. \ref{eq:PfSigmaTilde} and Eq. \ref{eq:PfSigmaT1Link}, we get 
	\begin{align*}
		\sigma^{\prime \prime}_{T^\prime}  \approx_{3 \sqrt{2\delta}} \tilde{\sigma}_{T^{\prime}}.
	\end{align*}
	Using this we can show that 
	\begin{align*}
		\eta_{RBT^\prime} & \approx_{5 \sqrt{2\delta} } (\Phi_R \otimes I_{(T \rightarrow B)})(\ket{\psi_{RT}} \bra{\psi_{RT}}) ) \otimes \tilde{\sigma}_{T^{\prime}}
	\end{align*}
	which proves our claim since $\ket{\psi_{RT}}$ was arbitrary.
\end{proof}

In the next Lemma, we consider the action of all the honest temporarily trusted nodes together during the protocol. This Lemma will be useful while proving the various security definitions. Let $E$ be the set of adversarial trusted nodes $(\TT_i)$. We assume that at least one of Alice and Bob is honest. This assumption is satisfied by the hypotheses of our security requirements. Under this assumption, during the commit stage Alice distributes the state $\rho_{A^\prime B_C T_{\bar{E}}T_{E}}$, where the register $B_C =(\Theta, Z, R)$ and is held by Bob. Apart from this restriction, we do not place any other restrictions on the behaviour of these parties. Alice, Bob and the adversarial trusted nodes $\TT_E$ follow an arbitrary (possibly dishonest) strategy during the protocol and the trusted nodes $\TT_{\bar{E}}$ act honest-but-curiously. Under these conditions, we show that the trusted nodes act almost honestly upto storing some fixed state in their memories.

\begin{lemma}
	Consider the erasable bit commitment protocol. Let $E$ be the set of adversarial trusted nodes $(\TT_i)$. We assume that one of the parties Alice or Bob is honest. Suppose that during the commit stage, (possibly dishonest) Alice prepares and distributes the state $\rho_{A^\prime B_C T_{\bar{E}}T_{E}}$, where the register $B_C = (\Theta, Z, R)$ and is held by Bob. Let $\Phi^O_{A B_C T_{E}}$ be the channel applied by Alice, Bob and the adversarial nodes in case of an 'Open'. Similarly, let $\Phi^E_{AB_C T_{E}}$ be the channel applied by Alice, Bob and the adversarial nodes in case of an 'Erase'. Then, in case of an 'Open', the state received by Bob $\eta_{A T^\prime_{\bar{E}} T^\prime_{E} B_C B}$ (where $T^\prime_{\bar{E}}$ and $T^\prime_{E}$ are registers held by $\TT_{\bar{E}}$ and $\TT_{E}$ after an 'Open') satisfies
	\begin{align*}
		\eta_{A^\prime T^\prime_{\bar{E}} T^\prime_{E} B_C B} \approx_{m \sqrt{2 \delta}} (I_{T_{\bar{E}} \rightarrow B_{\bar{E}}} \otimes \Phi^O_{AB_C T_{E}}) (\rho_{A^\prime B_C T_{\bar{E}} T_{E}}) \otimes \tilde{\sigma}_{T_{\bar{E}}}.
	\end{align*}
	for some fixed state $\tilde{\sigma}_{T_{\bar{E}}}$. In the case of an 'Erase', the state $\zeta_{A^\prime T^\prime_{\bar{E}} T^\prime_{E} B_C}$ received back by Alice satisfies (where $T^\prime_{\bar{E}}$ and $T^\prime_{E}$ are registers held by $\TT_{\bar{E}}$ and $\TT_{E}$ after an 'Erase')
	\begin{align*}
		\zeta_{A^\prime T^\prime_{\bar{E}} T^\prime_{E} B_C} \approx_{m \sqrt{2 \delta}} (I_{T_{\bar{E}} \rightarrow A_{\bar{E}}} \otimes \Phi^E_{AB_C T_{E}}) (\rho_{A^\prime B_C T_{\bar{E}} T_{E}}) \otimes \tilde{\omega}_{T_{\bar{E}}}.
	\end{align*}
	for some fixed state $\tilde{\omega}_{T_{\bar{E}}}$.
	\label{lemm:CollectiveHbC}
\end{lemma}
\begin{proof} Once again, we will prove the statement for the case of an 'Open'. The proof for an 'Erase' follows similarly. Suppose, that in case of an 'Open', each honest-but-curious node $T_i$ for $i \in \bar{E}$ applies the channel $\Phi_{T_i} := \Phi^O_{T_i \rightarrow B_i T^\prime_{i}}$ and the other parties apply the channel  $\Phi^O_{AB_C T_E}$ collectively to their states. Thus, 
\begin{align*}
	\eta_{A^\prime T^\prime_{\bar{E}} T^\prime_{E} B_C B} = \left( \left(  \bigotimes_{i \in \bar{E}} \Phi_{T_i} \right) \otimes \Phi^O_{AB_C T_{E}} \right) (\rho_{A^\prime B_C T_{\bar{E}} T_{E}})
\end{align*} 
Now using Lemma \ref{lemm:HbCDecouple}, we know that there exist fixed states $\tilde{\sigma}_{T_{i}}$ for $i \in \bar{E}$ for which  
\begin{align*}
	 \eta_{A^\prime T^\prime_{\bar{E}} T^\prime_{E} \Theta B}   & \approx_{5 \sqrt{2\delta}} \left( I_{T_{\bar{E}_1} \rightarrow B_{\bar{E}_1} } \otimes \left(  \bigotimes_{i \in \bar{E}; i \neq \bar{E}_1} \Phi_{T_i} \right) \otimes \Phi^O_{AB_C T_{E}} \right) (\rho_{A^\prime B_C T_{\bar{E}} T_{E}}) \otimes \tilde{\sigma}_{T_ {\bar{E}_1}} \\
	& \approx_{5 |\bar{E}| \sqrt{2\delta}} \left( I_{T_{\bar{E}} \rightarrow B_{\bar{E}} } \otimes \Phi^O_{AB_C T_{E}} \right)  (\rho_{A^\prime B_C T_{\bar{E}} T_{E}})\otimes  \bigotimes_{i \in \bar{E}} \tilde{\sigma}_{T^\prime_i} \\
	& \approx_{5 m \sqrt{2\delta}} \left( I_{T_{\bar{E}} \rightarrow B_{\bar{E}} } \otimes \Phi^O_{AB_C T_{E}} \right) (\rho_{A^\prime B_C T_{\bar{E}} T_{E}}) \otimes  \bigotimes_{i \in \bar{E}} \tilde{\sigma}_{T^\prime_i}
\end{align*}
where the second step is the result of repeating the first step multiple times, and in the third step we use the fact that $|\bar{E}| \leq m$. Hence, the Lemma is true for $\tilde{\sigma}_{T_{\bar{E}}} := \otimes_{i \in \bar{E}} \tilde{\sigma}_{T^\prime_i}$
\end{proof}
\subsubsection{Correctness for honest Alice and Bob}
\label{subsec:CorrHonAB}
\textbf{Correctness for Open: } 

\begin{proof} If Alice is honest, then she generates the state
\begin{align*}
	\rho_{X \Theta Z T} := \sum_{x,\theta, z} 2^{-(2n+k)} \ket{x} \bra{x} \otimes \ket{\theta } \bra{\theta} \otimes \ket{z} \bra{z} \otimes H^{\theta}\ket{\Enc(x)\oplus z} \bra{\Enc(x) \oplus z} H^{\theta}
\end{align*}
and sends $\theta, z$ to Bob, and the registers $T = (T_i)$ to the trusted nodes. Let $E$ be the set of adversaries. Suppose, that in case of an open each honest-but-curious node $T_i$ for $i \in \bar{E}$ applies the channel $\Phi_{T_i} := \Phi_{T_i \rightarrow B_i T^\prime_{i}}$ and the adversarial nodes apply the channel  $\Phi_{T_E} := \Phi_{T_E \rightarrow B_E T^\prime_E}$ collectively to their states. Let $\eta_{X \Theta B T'_{\bar{E}} T'_E}$ be the real state of the protocol after the trusted nodes forwards their shares to Bob, that is
\begin{align*}
	\eta_{X \Theta Z B T'_{\bar{E}} T'_E} = \left(\left(  \bigotimes_{i \in \bar{E}} \Phi_{T_i} \right) \otimes \Phi_{T_E} \right) (\rho_{X \Theta Z T})
\end{align*}
Further, we suppose that Alice and Bob shared a random $r \in \mathcal{R}$ for privacy amplification during the commit phase. By Lemma \ref{lemm:CollectiveHbC} we know that 
\begin{align}
	\eta_{X \Theta Z B T'_{\bar{E}} T'_E}  \approx_{5 m \sqrt{2\delta}} \left( I_{T_{\bar{E}} \rightarrow B_{\bar{E}} } \otimes \Phi_{T_E} \right) (\rho_{X \Theta Z T}) \otimes \tilde{\sigma}_{T^\prime_{\bar{E}}}.
	\label{eq:CorrReId}
\end{align}
for some fixed state $\tilde{\sigma}_{T^\prime_{\bar{E}}}$.
For this strategy of the trusted nodes we define the ideal state $\sigma_{CCFT}$ to simply be the state
\begin{align*}
	\sigma_{CCFT} := \sum_{c} 2^{-l} \ket{cc}\bra{cc} \otimes \ket{\suc}\bra{\suc} \otimes \sigma_T
\end{align*}
for $\sigma_T$ defined as 
\begin{align*}
	\sigma_T:= \tilde{\sigma}_{T^\prime_{\bar{E}}} \otimes \eta_{T^\prime_E}
\end{align*}
where $\eta_{T^\prime_E}$ is the partial state held by the nodes $T_E$ at the end of the 'Open' stage in the real protocol. This ideal state satisfies the security requirements given in \ref{def:ebc}. Further, if we apply the decoding map to both sides of Eq. \ref{eq:CorrReId}, we see that Bob is able to recover $\hat{X}= X$ with high probability since $|E| \leq t$ and the channel error rate is $\gamma$. Now observe that
\begin{align*}
	H^{5m\sqrt{2\delta}}_{\min} (X | T^\prime_{{\bar{E}}} T^\prime_{E})_{\eta} \geq (\delta_c+ \delta^\prime)n
\end{align*}
by using Eq. \ref{eq:CorrReId} and an argument similar to the one given in Sec. \ref{sec:hid-open-corr}. By privacy amplification ($\ell = \delta_c n$), we can show that the distance between the ideal state and the real state is at most
\begin{align*}
	\epsilon^\prime = 2^{-\frac{\delta^\prime n}{2}-1} +5m\sqrt{2 \delta}.
\end{align*}
\end{proof}

\textbf{Correctness for Erase: } Following the same arguments as above one can prove the correctness of Erase in our protocol. 

\subsubsection{Security for honest Alice}

The hiding properties of the protocol during the commit stage were proven in Section~\ref{sec:hid-commit}. \\

Further, the proof of security for Alice after an erase follows the argument as the proof for correctness of an Erase. The only difference being that the adversaries in this case have access to $\Theta$ as well. The arguments given in Section \ref{sec:hid-erase} handle this too. 

\subsubsection{Security for honest Bob}
\label{subsubsec:BindingProof}

\begin{proof}
We now formally prove security for honest Bob. We use a simulator argument similar to that used in Ref. \cite{KWW12} (proof of Theorem III.5) to define the ideal state $\sigma_{\tilde{C} A^\prime B T_E^\prime T_{\bar{E}}}$. We show that Alice is committed to $\tilde{C}$ after the commit phase, i.e. she cannot open something different than $\tilde{C}$ except with negligible probability.\\

Suppose that during the commit phase of the real protocol, Alice sends the $T$ part of the state $\rho_{A^\prime T }$ to the trusted nodes and $\theta$ and $r$ to Bob. Also, let $\Phi_{T_i}$ be the channel applied by the honest-but-curious node $i \in \bar{E}$. First, for the honest-but-curious party $i \in \bar{E}$, we let the state $\tilde{\sigma}_{T_i}$ be the state shown to exist in Lemma 2, such that for every initial state $\rho_{AT_{\bar{i}} T_i}$ and every channel $\Phi_{A T_{\bar{i}}}$ applied to the registers held by all the other parties
\begin{align*}
	(\Phi_{AT_{\bar{i}}} \otimes \Phi_{T_i}) (\rho_{AT_{\bar{i}} T_i}) \approx_{5\sqrt{2 \delta}} ( \Phi_{AT_{\bar{i}}} \otimes I_{T_i \rightarrow B_i} ) (\rho_{AT_{\bar{i}} T_i}) \otimes \tilde{\sigma}_{T_i^\prime}.
\end{align*}
Now, we show the existence of the ideal state $\sigma_{\tilde{C} A^\prime B T_E^\prime T_{\bar{E}}}$ algorithmically by making Alice and the $\TT_i$'s interact with a simulator. Imagine a protocol $\Pi_{\text{sim}}$ where instead of sending the quantum states to the trusted parties during the commit stage, Alice sends the shares of the honest-but-curious parties $\rho_{T_{\bar{E}}}$ to a simulator and the share of the adversaries to the adversaries. Alice also shares $\theta$ and $r$ with the simulator. The simulator measures $\rho_{T_{\bar{E}}}$ in the $\theta$ basis to get $\bar{y}_{\bar{E}}$. Define $\bar{y} := (\bar{y}_{\bar{E}}\ 0_E)$ and let $\tilde{y}$ be the codeword closest to $\bar{y}$. Further, let $\tilde{x} := \Dec (\tilde{y})$ and $\tilde{c} := \Ext(\tilde{x} , r)$. Then, it re-encodes $\bar{y}_{\bar{E}}$ in the $\theta$ basis (as $H^{\theta} \ket{\bar{y}_{\bar{E}}}$) and sends it to the honest-but-curious parties $\bar{E}$. Further, in this protocol, we assume that in the case of an 'Open', the honest-but-curious nodes $i \in \bar{E}$, $T_i$ apply the channel $\Phi^{(\text{id})}_{T_i}$ which we define as $\Phi^{(\text{id})}_{T_i} (\rho) = \rho \otimes \tilde{\sigma}_{T_i}$ for every state $\rho$. Finally, Alice and the adversaries $T_{E}$ apply whatever channel they apply during the commit stage in the real protocol. To prepare the ideal state $\sigma_{\tilde{C} A^\prime B T_E^\prime T_{\bar{E}}}$, the parties follow the above protocol. The register $\tilde{C}$ simply contains the string $\tilde{c}$.  \\ 

We now prove that for any strategy of Alice and the adversaries after the commit stage on the ideal state defined above, Bob opens the string $\tilde{c}$ during a successful open. Suppose, Bob measures the state, he receives during the open stage, in the $\theta$ basis and gets the outcome $\hat{y}$. Since, the $T_{\bar{E}}$ do not change or disturb their states during $\Pi_{\text{sim}}$, the Hamming distance ($h$) between $\hat{y}$ and $\bar{y}$ satisfies
\begin{align*}
	h(\hat{y}, \bar{y}) & \leq |E|\frac{n}{m} + \gamma n\\
	& \leq \left( \frac{t}{m} + \gamma \right) n = \frac{d-1}{4}.
\end{align*}
with high probability. Further, we have that $h(\hat{y}, \tilde{y}) \leq h(\hat{y}, \bar{y}) + h (\bar{y}, \tilde{y}) \leq 3 (d-1)/4$. For any codeword $y^\prime$ such that $y^\prime \neq \tilde{y}$, we have that 
\begin{align*}
	h(\hat{y}, y^\prime ) &\geq h (\tilde{y}, y^\prime) - h(\tilde{y}, \hat{y} ) \\
	& \geq d - \frac{3(d-1)}{4} > \frac{d-1}{4}.
\end{align*}
Since, Bob discards the commitment in case the Hamming distance between the measured string and the nearest codeword is more than $(d-1)/4$, the adveraries cannot change the commitment after the commit stage. Hence, the ideal state defined above satisfies the security definition.\\

Now, we will prove that the real state produced during the protocol is close to the ideal one after the open map is applied. Since, the action of the trusted nodes and the simulator in $\Pi_{\text{sim}}$ commutes, this protocol would produce the same state as a protocol where the simulator is between the trusted nodes and Bob. Then, observe that in $\Pi_{\text{sim}}$, since the simulator encodes the string $\bar{y}_{\bar{E}}$ in the $\theta$ basis and Bob measures this in the $\theta$ basis, we can simply remove this re-encoding and measurement step from the modified protocol. The protocol constructed this way, though, is exactly the original protocol, except for the fact that the honest-but-curious nodes instead of applying the maps $\Phi_{T_i}$, apply $\Phi^{(\text{id})}_{T_i}$. Using Lemma \ref{lemm:CollectiveHbC}, we can prove that the state of this protocol is $5m\sqrt{2 \delta}$ close to the real state. 
\end{proof}

\section{Comparison to other definitions for binding}

The binding condition used in the noisy storage model \cite{KWW12} requires the existence of a classical register $\tilde{C}$ extending the real state of the protocol $\rho_{ABT}$ to $\rho_{\tilde{C}ABT}$ such that the probability of successfully opening a commitment $\hat{C} \neq \tilde{C}$ on $\rho_{\tilde{C}ABT}$ is negligible. It has been argued in Ref. \cite{Kaniewski13}, that satisfying such a condition is not possible unless there is a register, inaccessible to the parties, which would prevent the purification of the protocol. As shown in Section \ref{sec:OTImpossible}, this is not the case for bit commitment protocols in the temporarily trusted party setting. Hence, similar to the relativistic setting, we cannot expect bit commitment protocols in the temporarily trusted party setting to satisfy this stronger bit commitment condition. \\

Several other papers \cite{Mayers97, Dumais00, Kent11, Kaniewski13} use a different weak-binding condition. For the statement of this condition, let us define $p_c$ to be the maximum probability of dishonest Alice being able to successfully open the commitment $c$. A bit commitment protocol is said to be $\Delta-$weak-binding according to the definition used in the aforementioned papers if
\begin{align}
	\sum_{c \in \{ 0,1 \}^\ell} p_c \leq 1 + \Delta.
	\label{eq:WeakBinding}
\end{align}
If a general protocol is binding according to our definition, then one can show that it is also $O(2^\ell \epsilon)$-weak binding, where $\epsilon$ is the security parameter in Definition \ref{def:ebc}. To see this note that for every state $\rho_{A'BT_{E}T_{\bar{E}}}$ at the end of the commit phase, there exists an ideal state $\sigma_{\tilde{C} A'BT_{E}T_{\bar{E}}}= \sum_{c } P_{\tilde{C}}(c) \ket{c}\bra{c} \otimes \sigma^{(c)}_{A'BT_{E}T_{\bar{E}}}$, with a classical distribution $P_{\tilde{C}}$ over $\{ 0,1 \}^\ell$. Alice can only open $\hat{C} \neq \tilde{C}$ with probability at most $\epsilon$ on this ideal state. To bound $p_c$ for every commitment $c$, one can choose the optimal open map allowing Alice to open the commitment $c$ with the maximum probability. By the security definition, this map would allow Alice to open $c$ with probability at most $P_{\tilde{C}}(c) + \epsilon$. Hence, the left hand side can be at most $2^\ell \epsilon$ more than 1 for the ideal state. One can similarly argue that the left hand side can be at most $O(2^\ell \epsilon)$ more than 1 for the real state. \\

We can, however, show that our protocol satisfies the weak binding condition mentioned above without the exponential dependence on $\ell$. This follows from the fact that in our protocol $(m-t)$ of the third parties are honest-but-curious and hence they do not interact with Alice, unless required by the protocol. As a result, the quantum channel applied by these parties (or the strategy adopted by them) is independent of $c$, which is chosen by Alice after the commit phase. The distribution over the commitment is determined completely by the shares held by the honest-but-curious parties and hence remains constant irrespective of the map applied by Alice and the corrupt third parties. Thus, for our protocol, the sum in Eq. \ref{eq:WeakBinding} can be bounded by one. \\

The $\Delta-$weak-binding condition presented above is known to be not composable (see e.g., \cite{DFRSS07, Kaniewski13} ). The binding condition used in the noisy storage model, on the other hand, is motivated through considerations for composability. As mentioned earlier, we adopted the binding condition in the noisy storage model for the temporarily trusted setting, so that composability results in the former may be lifted to this setting. 

\section{Decoupling of honest-but-curious nodes during the protocol}

In addition to satisfying the hiding property for Bob's and the adversarial nodes, our protocol satisfies the following hiding condition for the honest-but-curious nodes.

\begin{definition}[Hiding for the honest-but-curious nodes] If Alice is honest, $|E| \leq t$ and $\TT_{\bar{E}}$ are quantum honest-but-curious, then for any joint state $\rho_{C B^\prime T_E^\prime T_{\bar{E}}}$ created by the commit protocol, the honest-but-curious nodes do not learn $C$, or equivalently, they are completely decoupled from $C$:

$$\forall i \in \bar{E}: \rho_{C T_{i}} \approx_{\epsilon} \tau_{\ZO^l } \otimes \rho_{T_i}$$
\end{definition}

The proof that our protocol satisfies this property was given in Sec. \ref{sec:HbCDecoupling}. 

\section{Definitions and proofs for the expungement property}

\subsection{Definition of the expungement property}

\begin{definition}
\label{def:forgettingProp}

An $(\ell, \epsilon)$ randomized erasable string commitment scheme is said to satisfy the \emph{expungement on success} property if it satisfies the following properties:

Open: If both Alice and Bob are honest and $\sfA = \sfB = \suc$  after the open protocol, then there exists an ideal state $\sigma_{CCFT}$ satisfying the following

\begin{enumerate}

\item The distribution of $C$ is uniform and the $\TT_i$'s learn nothing about $c$:

$$\sigma_{CFT} = \tau_{\{ 0, 1 \}^\ell } \otimes \kb{\suc}{\suc} \otimes \sigma_T .$$

\item The joint state $\eta_{C \tilde{C} F T}$ created by the real protocol is $\epsilon$-close to the ideal state

$$ \eta_{C \tilde{C} F T} \approx_\epsilon \sigma_{CCFT} .$$

\end{enumerate}

Erase: If Alice is honest, Bob does not collaborate with the trusted nodes during the protocol and $\sfA =  erase$ after the erase protocol, then the joint state $\zeta_{C  B^{\prime } T }$ created by the erase protocol satisfies the following:

$$\zeta_{C  B^{\prime  } T } \approx_\epsilon \tau_{\{ 0, 1 \}^\ell } \otimes \zeta_{B^{\prime  } T } . $$

\end{definition}

\subsection{Proof of the expungement property}
\label{sec:ForgettingProof}

We only prove min-entropy bounds on the relevant states. These are sufficient to prove the required statements, by taking the commitment output length $\ell$ appropriately according to the privacy amplification theorem. However, depending on the noise parameters it might be necessary to change the some of the protocol parameters to achieve security as above. In order to complete the proof that our robust protocol satisfies the expungement property, we first list the properties of min-entropy that we will use, and then derive the desired min-entropy bound. We denote by $E$ the register of the adversarial trusted node who keeps information about $X$, by $\rho_{TXYZU \theta}$ the state prepared by Alice and by $\zeta_{EXYZU \theta} = \N_{T \rightarrow E} (\psi_{TXYZU \theta})  $ the state after the adversary acted on $T$ and kept register~$E$, after sending back to Alice what should have been the content of $T$ for erasure. We focus on proving the desired bound on $\Hmin^{7 \epsilon} (X|EZ)$ for $\zeta$ in the case of an erase (here $\epsilon$ is a small parameter chosen according to Ref. \cite{TLGR12} to achieve the bound given in Eq. \ref{eq:ucr}). A similar bound on $\Hmin^{7 \epsilon} (X|E)$ for $\eta$ in the case of an open  is proved similarly.

\subsubsection{Preliminaries for min-entropy bound}

We have $x \in_R \{0, 1 \}^k$, $z \in_R \{0, 1 \}^n$, $\theta \in_R \{0, 1 \}^n$ .

Also, $y = \Enc(x)$, $u = y \oplus z$, $y = u \oplus z$ and $\ket{\psi}_T = H^\theta \ket{u}_T$.

Hence, the partial states conditioned on $u$, $\rho_{TZ}^u = \rho_T^u \otimes \zeta_Z^u$ and then $\zeta_{EZ}^u = \N_{T \rightarrow E} (\rho_{TZ}^u) = \N_{T \rightarrow E} (\rho_{T}^u) \otimes \zeta_Z^u = \zeta_E^u \otimes \zeta_Z^u$, and then

\begin{align}
\label{eq:markov}
\Hmin (Z | UE) = \Hmin (Z|U).
\end{align}

Then, 

\begin{align}
\label{eq:hmineqk}
\Hmin (Z | U) = \Hmax (Z|U) = \Hmin (U | Z) = k,
\end{align}

since $\Hmin (Z | U) = \Hmin (Y | U) = \Hmin (Y) =\Hmin (X) = k$, and similarly for $\Hmax (Z | U)$ and $\Hmin (U | Z)$.

For smooth entropies, we also have that 

\begin{align}
\label{eq:function}
\Hmin^\epsilon (X | EZ) = \Hmin^\epsilon (Y|EZ) = \Hmin^\epsilon (U|EZ),
\end{align}

in which the last equality follows from \cite[Lemma~A.7]{DBWR14}.

We also make use of the following chain rules for smooth entropies (these are special cases of the inequalities proven in~\cite{VDTR13}),
with $f_\epsilon = \log ( \frac{1}{1 - \sqrt{1-\epsilon^2}} )$,

\begin{align}
\label{eq:cr1}
& \Hmin^{2\epsilon} (AB | C)  \geq \Hmin (A|BC) + \Hmin^{\epsilon} (B|C) - f_\epsilon , \\
\label{eq:cr2}
& \Hmax^{\epsilon} (AB | C)  \leq \Hmax (A|BC) + \Hmax (B|C) + f_\epsilon, \\
\label{eq:cr3}
& \Hmin^{2\epsilon} (AB | C)  \leq \Hmin^{7\epsilon} (A|BC) + \Hmax^{2\epsilon} (B|C) + 2 f_\epsilon, \\
\label{eq:cr4}
& \Hmax^{\epsilon} (AB | C)  \geq \Hmin (A|BC) + \Hmax^{2\epsilon} (B|C) - 2 f_\epsilon.
\end{align}

We use the fact that the smooth entropies satisfy a data processing inequality,

\begin{align}
\label{eq:dpi}
\Hmax^\epsilon (Z | E) \leq \Hmax^\epsilon (Z).
\end{align}

Finally, we use the following bound from uncertainty relation/sampling (see \cite{TLGR12}), for some parameters $\mu_\epsilon$ and $\delta_\epsilon$, and for $\gamma$ the tolerable noise rate,

\begin{align}
\label{eq:ucr}
\Hmin^\epsilon (U | E) \geq n (1 - H_2 (\gamma + \mu_\epsilon)) - \delta_\epsilon .
\end{align}

\subsubsection{Deriving the min-entropy bound}

We can now prove the desired bound on $\Hmin^{7 \epsilon} (X |EZ)$,

\begin{align*}
k - \Hmin^{7 \epsilon} (X|EZ) & = \Hmin (U|Z) - \Hmin^{7 \epsilon} (U|EZ)  &\textrm{by}~(\ref{eq:hmineqk})~\textrm{and}~ (\ref{eq:function}) \\
	& \leq \Hmax^{ \epsilon} (UZ) - \Hmax^{2 \epsilon} (Z) + 2 f_\epsilon  &\textrm{by}~(\ref{eq:cr4}) \\
	& \quad - \Hmin^{2 \epsilon} (UZ|E) + \Hmax^{2 \epsilon} (Z|E) + 2f_\epsilon  & \textrm{by}~(\ref{eq:cr3}) \\
	& \leq \Hmax^{ \epsilon} (UZ) - \Hmin^{2 \epsilon} (UZ|E) + 4 f_\epsilon  & \textrm{by}~(\ref{eq:dpi}) \\
	& \leq \Hmax (U) + \Hmax (Z|U) + f_\epsilon  & \textrm{by}~(\ref{eq:cr2}) \\
	& \quad - \Hmin^{ \epsilon} (U|E) - \Hmin (Z|UE) + f_\epsilon + 4 f_\epsilon  & \textrm{by}~(\ref{eq:cr1}) \\
	& = n + \Hmax (Z|U) + 6 f_\epsilon  \\
	& \quad - \Hmin^{ \epsilon} (U|E) - \Hmin (Z|U)   & \textrm{by}~(\ref{eq:markov}) \\
	& = n - \Hmin^{ \epsilon} (U|E) + 6 f_\epsilon  & \textrm{by}~(\ref{eq:hmineqk}) \\
	& \leq n - n (1 - H_2 (\gamma + \mu_\epsilon)) + \delta_\epsilon + 6 f_\epsilon  & \textrm{by}~(\ref{eq:ucr}) \\
	& = n ( H_2 (\gamma + \mu_\epsilon)) + \delta_\epsilon + 6 f_\epsilon.
\end{align*}

Rearranging terms, we get

\begin{align}
\Hmin^{ 7 \epsilon} (X|EZ) \geq k - n ( H_2 (\gamma + \mu_\epsilon)) - \delta_\epsilon - 6 f_\epsilon.
\end{align}

\end{document}